\documentclass[journal,draftclsnofoot,onecolumn,12pt,twoside]{IEEEtranTCOM}
\usepackage[dvipsnames]{xcolor}

\usepackage[skip=-5 pt]{caption}
\usepackage{tikz}
\usepackage{pgfplots}
\usepackage{booktabs}
\usepackage{subcaption}

\usepackage{amsmath,epsfig,amssymb,algorithm,algpseudocode,amsthm,cite,url}
\usepackage[shrink=30]{microtype}
    \DeclareMathOperator{\tr}{tr}

\allowdisplaybreaks 
\theoremstyle{plain}
\usepackage{empheq}
\newtheorem{theorem}{Theorem}
\newtheorem{assumption}{Assumption}
\usepackage{cleveref}

\newtheorem{lemma}{Lemma}

\newtheorem{cor}{Corollary}
\theoremstyle{definition}

\theoremstyle{remark}
\newtheorem*{rem}{Remark}

\newtheoremstyle{specialcasestyle}{1mm}{1mm}{\upshape}{}{\bfseries\upshape}{.}{0mm}{}
\theoremstyle{specialcasestyle}

\newcommand{\bPhi}{\boldsymbol{\Phi}}

\newcommand{\bxi}{\boldsymbol{\xi}}

\newcommand{\bR}{{\bf R}}

  \newcommand{\bA}{{\bf A}}
  \newcommand{\bB}{{\bf B}}

    \newcommand{\bI}{{\bf I}}
      \newcommand{\bh}{{\bf h}}
      \newcommand{\bH}{{\bf H}}

      \newcommand{\wbH}{\widehat{\bf H}}
      \newcommand{\wbHh}{\widehat{\bf H}^{\mbox{\tiny H}}}
      
      \newcommand{\wbh}{\widehat{\bf h}}
      \newcommand{\wbhh}{\widehat{\bf h}^{\mbox{\tiny H}}}

      \newcommand{\bx}{{\bf x}}

          \newcommand{\bz}{{\bf z}}
\newcommand{\bg}{{\bf g}}
        
        \newcommand{\bq}{{\bf q}}
      
      \newcommand{\bQ}{{\bf Q}}

       \newcommand{\herm}{^{\mbox{\tiny H}}}

\usetikzlibrary{arrows.meta}

\tikzset{
    myarrow/.style={-{Triangle[length=2mm,width=1.4mm]}}
}

\newcommand{\height}{7.0cm}
\usepackage{epstopdf}
\usepackage{graphicx}
\usepackage{color}
\usepackage{placeins}
\usepackage{float}
\usepackage{tabularx,colortbl}
\newcommand{\asto}{{\xrightarrow[N\to\infty]{a.s.}} \hspace{0.1cm}}

\begin{document}
\title{ LMMSE Receivers in Uplink Massive MIMO Systems with Correlated Rician Fading}
\author{\small \authorblockN{Ikram Boukhedimi, Abla Kammoun, and Mohamed-Slim Alouini}\\
\authorblockA{Computer, Electrical, and Mathematical Sciences and Engineering Division,\\
King Abdullah University of Science and Technology (KAUST),\\Thuwal, Makkah Province, Kingdom of Saudi Arabia.\\Email: \{ikram.boukhedimi, abla.kammoun, slim.alouini\}@kaust.edu.sa}}
\maketitle
\vspace{-5 em}
%%%%%%%%%%%%%%%%%%%%%%%%%%%%%%%%%%%%%%%%%%%%%%%%%%%%%%%%%%%%%%%%%%%%%%%%%%%%%5
\begin{abstract}\vspace{-1em}
{We carry out a theoretical analysis of the uplink (UL) of a massive MIMO system  with per-user channel correlation and Rician fading, using two processing approaches. Firstly, we examine the linear-minimum-mean-square-error receiver under training-based imperfect channel estimates. Secondly, we propose a statistical combining technique that is more suitable in environments with strong Line-of-Sight (LoS) components. We derive closed-form asymptotic approximations of the UL spectral efficiency (SE) attained by each combining scheme in single and multi-cell settings, as a function of the system parameters. These expressions are insightful in how different factors such as LoS propagation conditions and pilot contamination impact the overall system performance. Furthermore, they are exploited to determine the optimal number of training symbols which is shown to be of significant interest at low Rician factors. The study and numerical results substantiate that stronger LoS signals lead to better performances, and under such conditions, the statistical combining entails higher SE gains than the conventional receiver. }
%Furthermore, they are also exploited to find the optimal number of training symbols and to conduct a comparative analysis between the processing techniques. Although our analytical results are based on a large antenna-limit, we show by simulations that they provide very accurate approximations even for finite system dimensions.
 %The latter estimates the Rician correlated channels through training symbols and uses linear-minimum-mean-square-error for signal processing. In addition to the consumption of the coherence length, channel estimation is subject to errors which degrades the performances. In this line, we investigate the possibility of avoiding channel estimation by using a ``statistical" receiver that relies only on the statistics of the system. Assuming the asymptotic $N$ regime with a fixed $K$, we derive closed-form approximations of the achievable UL rates for both conventional and statistical receivers as a function of the system's parameters. These expressions are instructive in how these parameters impact the performances, and also harnessed to carry out a comparative analysis between the above processing techniques. Although  our  analytical results are based on a large antenna-limit, we show by simulations that they provide very accurate approximations even for finite system dimensions.
%Plus, a different processing approach using only the statistics of the channels is investigated and found to be a judicious choice when the Rician factor is high enough.  
\end{abstract}
\vspace{-1.5em}
\begin{IEEEkeywords}\vspace{-1em}
\hspace{-3em} \small Massive MIMO, correlated Rician fading, imperfect channel estimation, optimal training, LMMSE combining.
\end{IEEEkeywords}
\vspace{-2em}
\section{Introduction}\vspace{-0.5em}
 \par Future 5G networks are expected to support substantial amounts of mobile data traffic and ensure a reliable quality of service to the end users\cite{Noncooperative-Marzetta2010,5G-Andrews2014,MMIMO-Lu2014}. Over the past couple of years, massive MIMO has been extensively investigated, (see \cite{Noncooperative-Marzetta2010,5G-Andrews2014,MMIMO-Lu2014,HowMany-Jacob2013,Large-Wagner2012, lowcomplexity-abla2014,Coordinated-Ikram2017} and references therein). In general, most of these works consider Rayleigh fading channels that model scattered signals, yet do not encompass the possibility of a Line-of-Sight (LoS) component which is commonly present in practical wireless propagation scenarios and modeled by Rician-fading. At the same time, in order to meet the 5G performance demands, massive MIMO is expected to be omnipresent and thus, all propagation conditions ought to be examined. This is the case  for indoor applications or small areas operating over mmWave communications wherein the presence of LoS is conceivable \cite{MIMOmmWave-Sun2014}.  
 %. In this context, it is relevant to note that propagation in mmWave frequencies does not obey the Rayleigh-fading model but rather tends towards a Line-of-Sight or near-LoS environment, which is modeled by Rician-fading .
In fact, active research is being conducted to study the pairing of massive MIMO systems with mmWave communications to jointly reap their benefits \cite{Massive-Bogale2016,mmWave-Roh2014,Millimiter-Swindle2014}.
Accordingly, understanding the performance of massive MIMO systems operating under LoS conditions is of growing importance and is the focus of the present work.

\par { The main literature related to this work is represented by \cite{Power-Zhang2014,Achievable-Tan2015,Performance-Tataria2016,Performance-Kong2015,asymptotic-falconet2016,LowComp-Ghasham2017,Achievable-Zhang2016,Asymptotic-Luca2017,Impact-tataria2017,Optimal-Jacob2011, Area-Nyarko2017, Globecom-Ikram2018, Downlink-Dong2018,Beamforming-Yue2015,MassiveMIMO-Ozgecan2018}. The authors in \cite{Power-Zhang2014} investigate the power scaling laws and the uplink (UL) rates using zero-forcing (ZF) and maximum-ratio (MR) combiners, however assuming uncorrelated channels. In \cite{Performance-Tataria2016}, an analytical study of the rates of the downlink (DL) of a MU-MIMO system assuming ZF beamforming and uncorrelated Rician fading channels is performed. Conversely, in \cite{asymptotic-falconet2016}, the authors use tools from random matrix theory (RMT) to conduct an asymptotic analysis of the DL of a spatially correlated MIMO Rician fading model, yet under the assumption of perfect channel state information at the base station (BS). A similar large system analysis relying on some recent RMT results is led in \cite{Asymptotic-Luca2017}. In this latter, the authors investigate regularized-ZF and MR-transmit schemes in the DL of a large-scale MIMO system under uncorrelated Rician fading system and assuming imperfect channel estimates. Moreover, in the same line, the authors in \cite{Optimal-Jacob2011} use similar tools to analyze the ergodic UL rates of a massive MIMO system and determine the optimal fraction of the coherence time used for channel training in a Rayleigh fading setting. The work in \cite{Beamforming-Yue2015} analyzes transmit and receive LoS-based beamforming designs which treat the scattered signals as interference in Rician fading massive MIMO systems. { Finally, the work \cite{MassiveMIMO-Ozgecan2018} recently appeared after our submission, and it focuses on correlated Rician massive MIMO systems using the less involved MR combining and precoding scheme.
{\let\thefootnote\relax\footnote{A part of this paper has been presented at the IEEE ICC18, May 2018, Kansas City, MO, USA.}}

\par In this paper, we investigate the UL performance of a massive MIMO system wherein every user is allotted a distinct channel correlation matrix and Rician factor, and above all, assuming imperfect CSI, and ultimately, pilot contamination. Furthermore, we consider two different combining schemes. The first method consists in the involved  Linear-Minimum-Mean-Square-Error (LMMSE) receiver which we will refer to as the `conventional' receiver in the sequel. Note that LMMSE's performance underlying such an intricate system model render the analysis comprehensive and unprecedented as, to the best of our knowledge, it has not been conducted thus far.  As to the second technique, we propose a `statistical' combiner that only utilizes the long-term statistics of the channels like the Rician factors, spatial correlation matrices, etc. In essence, this technique is purposely designed for LoS-prevailing environments to circumvent training and channel estimation and associated errors, and exploit the presence of LoS components in a more efficient manner. \\
{We first consider a single-cell network with the aforementioned comprehensive channel model,} and carry out a theoretical analysis of the achievable UL spectral efficiency (SE) for each receiver. Assuming the large antenna-limit with a fixed number of users, we harness some rudimentary asymptotic tools such as the law of large numbers (LLN) and convergence of quadratic forms \cite{eigenvaluesoutside-silverstein2009}, to derive closed-form approximations of the SEs. These approximations provide insights on the impact of the system parameters such as the training sequence's length, the Rician factor as well as the propagation conditions on the overall performances. Furthermore, we exploit them to determine the optimal number of training symbols that maximizes the UL SE whose value is shown to be  particularly crucial at low Rician factors. A relevant outcome of the study reveals that high Rician factors generate far better performances, and that in such environments, longer training sequences are rather counterproductive since they degrade the achievable SEs. This result, led us to propose the statistical receiver for systems with strong specular signals. As will be elaborated later, this scheme is obtained through the maximization of the corresponding UL SE and proven to outperform the conventional receiver in LoS-prevailing systems. { Note that in the conference paper \cite{Avergae-Ikram2017}, we present the preliminary results of this single-cell setting.} %Nevertheless, in this paper, this result is proved and examined in copious detail and thus, further insights are provided.
\par For a more thorough analysis, we extend our study to a multi-cell system in order to examine the performances of both combining schemes when they are subject to inter-cell interference, and especially to pilot contamination. %Therefore, we consider a multi-cell network with 
Similarly to the singe-cell, we derive closed-form approximations of the achievable SEs under the asymptotic antenna regime. Accordingly, the impacts of LoS propagation conditions and pilot-contamination-induced interference are meticulously analyzed for both receivers. Ultimately, the discussion highlights interesting aspects on the interplay between the Rician factors and pilot contamination. Additionally, it  unveils that in a multi-cell setup, the proposed combiner outperforms the conventional one to an even higher extent than it is the case for a single-cell system. Evidently, numerical results are provided to validate the accuracy of our analytical findings and better illustrate the efficiency of both schemes in all settings for finite system dimensions, albeit computed in the asymptotic regime.} 
\par The remainder of the paper is organized as follows. Section \ref{sec:single} encompasses the single-cell system starting by the corresponding UL system model, followed by the detection schemes, and finally the theoretical analysis of the achievable spectral efficiencies that shed light on some interesting aspects. Then, pursuing a similar rationale as in Section \ref{sec:single}, Section \ref{sec:multi} focuses on the multi-cell systems. After that, Section \ref{sec:Num Results} consists of a selection of numerical results that confirm the theoretical derivations given in Sections \ref{sec:single} and \ref{sec:multi}. Finally, conclusions are drawn in Section\ref{sec:Conclusion}.
\vspace{-1.3em}
\section{Performance Analysis in a Single-Cell Setting}\label{sec:single}\vspace{-0.5em}
%\subsection{System model}\label{sec: System Model}
We consider uplink transmissions of a TDD single-cell system with $K$ mono-antenna users (UEs) and a BS equipped with $N$ antennas. Assuming Gaussian codebooks, the vector of the transmitted data symbols sent by all the UEs is denoted $\sqrt{\frac{p_u}{N}}\bx \sim\mathcal{CN}(0,\frac{p_u}{N} \bI_K)$ and therefore, the received signal at the BS writes{\let\thefootnote\relax \footnote{\small \it Notations: In the sequel, bold upper and lowercase characters refer to matrices and vectors, respectively. We also use $(.)^{\text{H}}$, $\tr(.)$ and $(.)^{-1}$ to denote the conjugate transpose, the trace of a matrix and the inverse operations, respectively. $log(.)$ is the natural logarithm, the $N \times N$ identity matrix is denoted $\textbf{I}_{N}$, and $\delta_{j\ell}$ is the Kronecker delta. Finally, $\left[{\bf A} \right]_{ij}$ is the element on the $i-$th row and $j-$th column of matrix $\bf A$ and $diag\left\{ a_i \right\}_{i=1}^{N}$ is the $N\times N$ diagonal matrix with $a_i$ being its $i-$th diagonal element.}}:
\vspace{-0.5em}\begin{equation}
{\bf y}=\sqrt{\frac{p_u}{N}} \bH \bx +{\bf n},
\label{eq:y sig_received}
\end{equation}
where $\bH=\left[\bh_1,\dots, \bh_K \right]$ is the $N\times K$ aggregated MIMO channel matrix from all UEs to the BS and ${\bf n}$ represents a zero-mean additive Gaussian noise with variance $\sigma^2$. Correlated Rician fading channels are considered such that the channel between the $k-$th UE and the BS is modeled as: \vspace{-0.5em}
\begin{equation}
\bh_k= \sqrt{\beta_k} \left(\sqrt{\frac{1}{1+\kappa_k}}\boldsymbol{\Theta}_k^{\frac{1}{2}}\bz_k + \sqrt{\frac{\kappa_k}{1+\kappa_k}}\overline{\bz}_k\right),
\label{eq:channel model}
\end{equation}
where $\beta_k$ accounts for the large-scale channel fading of UE$_k$ and the second term represents the small-scale fading channel. This latter consists of the Rayleigh component $\bz_k \sim \mathcal{CN}\left(0,\bI_N\right)$ to depict scattered or Non-LoS signals and the deterministic component $\overline{\bz}_k\in\mathbb{C}^{N}$ to represent the specular (LoS) signals. For each UE, the ratio between these components is depicted by the Rician factor $\kappa_k$. Plus, for each UE $k$, we consider a different channel correlation matrix $\boldsymbol{\Theta}_k$. Throughout the paper, $\forall k, \ \boldsymbol{\Theta}_k$ is assumed to be slowly varying compared to the channel coherence time and thus is supposed to be perfectly known to the BS.
Finally, for notational convenience, we let: $\overline{\bh}_k =\sqrt{\frac{\beta_k\kappa_k}{1+\kappa_k}}\overline{\bz}_k$ and ${\bf R}_k = \frac{\beta_k}{1+\kappa_k}\boldsymbol{\Theta}_k$. Therefore, $\bh_k \sim \mathcal{CN}\left(\overline{\bh}_k,{\bf R}_k\right)$.
\vspace{-1 em}
\subsection{Channel Estimation}
%\begin{figure}[t]
%\centering
%\includegraphics[scale=0.27,natwidth=610,natheight=642]{Publication11.eps}
%\vspace*{0.5 em}
%\caption{\small UL of block-fading system of coherence length $T$ symbols with pilot/data transmissions scheme.}
\vspace*{-0.5 em}
%\label{fig: coherence block}
%\end{figure}
In practice, prior to processing the received signal, the BS estimates the channel matrix $\bH$. Let $\wbH=\left[\wbh_1,\dots,\wbh_K\right]$ denote the aggregate matrix of these estimates. In TDD systems, each UL channel coherence block of length $T$  is split into two phases starting by training and followed with data transmission%, \textit{(See Fig.\ref{fig: coherence block})}
. In the pilot training interval of $\tau\geq K$ symbols%\footnote{ The orthogonality of the pilot sequences imposes $\tau \geq K$.}
, all $K$ UEs broadcast orthogonal sequences of known pilot symbols with average power $\tau p_p$. It is important to note that in the considered Rician fading, since the specular components are hardly changing, it is reasonable to assume that both the LoS component and Rician factors of all UEs are known to both the transmitter and receiver. Accordingly, using single-cell LMMSE estimation, the estimate ${\wbh}_k$ of the channel $\bh_k$  is given by \cite{HowMany-Jacob2013} : 
\begin{equation}
{\wbh}_k= {\bf R}_k \bPhi_k \left({\bf y}^{tr}_k-\overline{\bh}_k\right)+ \overline{\bh}_k,
\label{eq:h_hat}
\end{equation}
where, 
%\begin{align}
%\bPhi_k &=\left({\bf R}_k + \frac{1}{\rho_{tr}}\bI_N\right)^{-1}, \label{eq:Phi_k}\\
%{\bf y}_k^{tr} & = \bh_k + \frac{1}{\sqrt{\rho_{tr}}}n_k^{tr}, \label{eq:y_tr}
%\end{align}
%\begin{equation}
$\bPhi_k =\left({\bf R}_k + \frac{1}{\tau\rho_{tr}}\bI_N\right)^{-1}, \quad {\bf y}_k^{tr} = \bh_k + \frac{1}{\sqrt{\tau\rho_{tr}}}n_k^{tr},$ %\label{eq:y_tr}
%\end{equation}
 and $\rho_{tr} ={\dfrac{p_p}{\sigma^2}}$ is the SNR corresponding to the training phase.{ The higher value $\tau\rho_{tr}$ takes, the better quality of channel estimation becomes. In fact, as $\tau\rho_{tr} \rightarrow \infty, \ \wbH \rightarrow \bH$ which corresponds to the perfect CSI scenario.} From \eqref{eq:h_hat}, it can be shown that $\wbh_k \sim \mathcal{CN} \left(\overline{\bh}_k, \tilde{\bf R}_k\right),$ with $\tilde{\bf R}_k={\bf R}_k \bPhi_k{\bf R}_k$. Plus, considering the orthogonality property of LMMSE estimation, the estimation error $\bxi_k = \bh_k-\wbh_k$ follows the distribution $\mathcal{CN}\left(0, {\bf R}_k- \tilde{\bf R}_k\right)$.
\vspace{-1em}
\subsection{Detection and Achievable Uplink Spectral Efficiency}
To process the signal $\bf y$ \eqref{eq:y sig_received}, the BS uses a linear receiver. In this work, we are interested in the conventional LMMSE receiver that relies on acquired channel estimates, and propose a statistical receiver that is mainly based on the long-term parameters of the system. % we lead a comparative analysis of the conventional LMMSE receiver based on channel estimates and a proposed statistical receiver that maximizes the UL-SINR.
\subsubsection{Conventional LMMSE Receiver}
Let $\bg_k$ %${\bf G}=\left[\bg_1 \ \bg_2 \dots \bg_K\right] \in 
$\mathbb{C}^{N\times 1}$ denote the  conventional combining vector used to process the single sent by UE $k$. Under imperfect channel estimation conditions, $\bg_k$ is defined as \cite{book-Kay97,HowMany-Jacob2013}:
\begin{equation}
{\bg_k} = \left({\wbH\wbHh} + \sum_{i=1}^{K} \left({\bf R}_i-\tilde{\bf R}_i\right) + \frac{N}{\rho_d} \bI_N\right)^{-1}\wbh_k. \label{eq:G_MMSE_hat}
\end{equation}
In order to retrieve useful data, the BS generates the signal ${\bf r}={\bf G}^{\mbox{\tiny H}} {\bf y}.$ 
As shown in \eqref{eq:G_MMSE_hat}, the design of $\bg_k$ leverages the channel estimates $\wbH$, thus making the performances sensitive to channel estimation errors. Therefore, the $k$-th element of $\bf r$ can be decomposed as:
 \begin{equation}
{r}_k= %\sqrt{\frac{p_u}{N}} \bg\herm_k \wbh_k x_k + 
\sqrt{\frac{p_u}{N}}
 {\sum_{\substack{i=1}}^{K} \bg\herm_k \wbh_i x_i} + \sqrt{\frac{p_u}{N}} \sum_{i=1}^{K} \bg\herm_k {\bf \bxi}_i x_i +  \bg\herm_k n_k,
\label{eq:decomposed r}
\end{equation}
This expression respectively, separates the signal and intra-cell interference, channel estimation errors and noise terms. Additionally, when a pre-training phase of $\tau$ symbols is performed, only a fraction of the total coherence block is used for useful data transmission. Therefore, denoting $\rho_d=\frac{p_u}{\sigma^2}$, the UL achievable SE for UE $k$, in case of channel-estimate based conventional processing is defined as \cite{fundamentals-marzetta2016}: \footnote{In the sequel, we add the superscripts $(.)^{\rm conv, S}$, $(.)^{\rm stat, S}$, $(.)^{\rm conv, M}$ and $(.)^{\rm stat, M}$ to, respectively, distinguish the relevant quantities corresponding to conventional combining and statistical combining, in single-cell and multi-cell schemes.} : 
{%\small 
\begin{align}
& {\rm SE}_k^{\rm conv,S}=\left(1-\frac{\tau}{T}\right)  %\notag \\ &
\mathbb{E}\Bigg[\log\bigg(1+ \frac{\vert{\bf g}_k\herm\wbh_k\vert^2}
{\mathbb{E}\big[{\sum_{\substack{i=1\\ i\neq k}}^{K}\vert{\bf g}_k\herm{\wbh}_i\vert^2}+{\sum_{\substack{i=1}}^{K}\vert{\bf g}_k\herm{\bf \bxi}_i\vert^2}+\frac{N}{\rho_d}\Vert{\bf g}_k\Vert^2\big]}\bigg)\Bigg]. 
\label{eq:sum rate hat}
\end{align} }
\vspace{-1em}
\subsubsection{Statistical LMMSE Receiver}
Due to its slow varying pace, the LoS component $\overline{\bH}=\left[\overline{\bh}_1,\dots,\overline{\bh}_K\right]$ can be easily estimated. For example, the BS may estimate the specular signals in a previous transmission from the UEs, in contrast to the Rayleigh signals which must be estimated at every $T$. In addition, choosing the right number of training symbols, $\tau$, is paramount to ensure the overall UL performances since a small $\tau$ entails significant estimation errors and a larger $\tau$ suggests less transmitted data. { Motivated by these factors, we propose in this work a statistical receiver denoted $\overline{\bg}_k$, that exclusively exploits the presence of the quasi-deterministic LoS component $\overline{\bH}$ and the long-term parameters of the system, such as the spatial correlation matrices ${\bf R}_k$, the large-scale fading factors $\beta_k$, etc. Naturally, using such a receiver enables to avoid training and channel estimation altogether, thereby yielding the single-cell UL SE ${\rm SE}_k^{\rm stat,S}$}:
\begin{equation}\label{eq:sum rate stat}
{\rm SE}_k^{\rm stat,S}= \mathbb{E}\left[\log\left(1+ \frac{\vert\overline{\bf g}_k\herm\overline{\bh}_k\vert^2}{\mathbb{E}\big[\overline{\bf g}_k\herm\left(\sum_{\substack{i=1}}^{K}{\bh}_i{\bh}\herm_i-\overline{\bh}_k\overline{\bh}_k\herm+\frac{N}{\rho_d}\bI_N\right)\overline{\bf g}_k\herm \big] }\right) \right].
\end{equation}
We propose to design $\overline{\bg}_k$ through the maximization of 
%$\overline{\rm SE}_k^{\rm stat}$. This latter, is 
a deterministic equivalent of ${\rm SE}_k^{\rm stat,S}$ in the infinite antenna limit which we denote $\overline{\rm SE}_k^{\rm stat,S}$. Specifically, for $k=1,\dots, K$, $\overline{\bg}_k$ is defined as: \vspace{-1em}
\begin{align*}\label{P1}
&\overline{\bg}_k=  \ {\underset{\overline{\bg}_k}{\text{argmax}}} \ \overline{\rm SE}_k^{\rm stat}, \tag{P1} \\
& \text{s.t} \quad {\rm SE}_k^{\rm stat,S}-\overline{\rm SE}_k^{\rm stat,S} \asto 0.
\end{align*}
As shall be seen in the next section, on account that $\overline{\bg}_k$ is deterministic, one should note that $\overline{\rm SE}_k^{\rm stat,S}$ is obtained by means of quite rudimentary asymptotic tools. 
 \vspace{-2.5em}
\subsection{ Asymptotic Analysis of the Single-Cell Performances }\label{sec:Analysis}\vspace{-0.5em}
In this section, we carry out a comparative theoretical analysis between the UL performances achieved by the conventional  receiver, $\bg_k$ \eqref{eq:G_MMSE_hat}, and the proposed statistical combiner, $\overline{\bg}_k$ \eqref{P1}. The study is conducted under the assumption of imperfect channel state information, and a distinct Rician factor as well as channel correlation per user. Ultimately, the objective is to  determine conditions in which the statistical receiver outperforms the conventional one. Nonetheless, as can be seen from ${\rm SE}_k^{\rm conv,S}$ \eqref{eq:sum rate hat} and ${\rm SE}_k^{\rm stat,S}$ \eqref{eq:sum rate stat}, these expressions involve random quantities that are rather compact and do not lend themselves to simple interpretations nor manipulations. { Accordingly, we first derive closed-form asymptotic approximations of both ${\rm SE}_k^{\rm conv,S}$ and ${\rm SE}_k^{\rm stat,S}$ which we exploit thereafter for the comparison. To obtain these approximations, we consider the large-antenna limit with a fixed number of UEs. This can be formulated as:}
\vspace{-1.2em}
\begin{assumption} \label{ass:asymptotic}
We assume that $K$ is fixed while $N$ grows large without bound. We also consider that as $N \rightarrow \infty$, $\forall k$, the channel correlation matrix has a bounded spectral norm $\Vert \boldsymbol{\Theta}_k\Vert_2$ . For simplicity, this asymptotic regime will be denoted by $N\rightarrow \infty$. 
\end{assumption}\vspace{-1em}
\subsubsection{Conventional Combining in Single-Cell Systems}\vspace{-0.5em}
Define the matrices:\vspace{-0.6em}
{\begin{align}
&{\bf Q}=\left(\frac{1}{N}\overline{\bH}\herm\overline{\bH}+\frac{1}{N} diag\left\{\tr \tilde{\bf R}_\ell \right\}_{\ell=1}^{K}+\frac{1}{\rho_d}\bI_K\right)^{-1}, \label{eq:Q_hat} \\
& {\bf T}_i=\overline{\bH}\herm \frac{1}{\tau\rho_{tr}} {\bf R}_i \bPhi_i \overline{\bH} + diag \left\{\tr(\tilde{\bf R}_\ell \frac{1}{\tau\rho_{tr}}{\bf R}_i\bPhi_i)\right\}_{\ell=1}^{K},
\end{align}}\vspace{-0.5em}
and let ${\bq}_k$ be the $k-$th column of the matrix ${\bf Q}$.
\begin{theorem}[Conventional combining in single-cell systems]\label{th: MMSE hat}
 Under Assumption \ref{ass:asymptotic}, we have : $ {\rm SE}_k^{\rm conv,S} - \overline{\rm SE}_k^{\rm conv,S} \asto 0 $,  such that 
 %$\overline{\rm SE}_k^{\rm conv,S}$ is given by \eqref{eq:R MMSE DE} : %({in the top of the next page}). 
%\begin{figure*}[t]
%%%\vspace{-0.6cm}
\end{theorem}
{\begin{align}
\overline{\rm SE}_k&^{\rm conv,S} =\left(1-\frac{\tau}{T}\right) % \notag \\ &
 \log\left(1+ 
\frac{\left|1- \frac{1}{\rho_d}\left[{\bf Q}\right]_{kk}\right|^2}{ {\sum_{\substack{i=1\\ i\neq k}}^{K}}\left|\frac{1}{\rho_d}\left[{\bf Q}\right]_{ki}\right|^2
+ \frac{1}{N^2} {\sum_{\substack{i=1}}^{K}} {\bq}_k\herm  {\bf T}_i {\bq}_k
+ \frac{1}{\rho_d}\left(\left[{\bf Q}\right]_{kk} -\frac{1}{\rho_d} \left[{\bf Q}^2\right]_{kk}\right)
 }\right). \label{eq:R MMSE DE}
\end{align}}
%\hrulefill
%%\vspace{-0.6cm}
%\end{figure*}
\vspace{-0.5em}\begin{proof}A proof is given in Appendix \ref{app: conv LMMSE}
\end{proof}
% For instance, we can see that as $N$ grows infinitely large, { $\frac{1}{N^2} \tr\left((\tilde{\bf R}_k+ \overline{\bh}_k \overline{\bh}\herm_k)\bPhi_i{\bf R}_i\right) \asto 0$} which implies that the LMMSE channel estimation error vanishes in the UL massive MIMO setting. In the same line, $ \frac{1}{N^2} \tr\left(\tilde{\bf R}_k ( \overline{\bh}_i\overline{\bh}_i\herm +\tilde{\bf  R}_i ) \right) \asto 0$, revealing that for massive MIMO Rician fading channels, only the specular components cause intra-cell interference which is characterized by the inner product between $\overline{\bh}_k$ and $\overline{\bh}_i$. Therefore, $\overline{\rm SE}_{\rm MRC}$ can be further simplified to : 
We provide in the closed-form expression \eqref{eq:R MMSE DE} approximations of all the different terms constituting ${\rm SE}_k^{\rm conv,S}$. This allows to have some insights on the behavior of these signals and their impact on the achievable SE. Note nonetheless that further simplifications can be made in the infinite antenna limit. For instance, we can see that as $N$ grows infinitely large and for a fixed $K$ : {$ \frac{1}{N^2} \sum_{i=1}^{K} {\bq}_k\herm  {\bf T_i} {\bq}_k \asto 0$}, therefore implying that channel estimation errors vanish in the UL massive MIMO setting. Accordingly, $\overline{\rm SE}_k^{\rm conv,S}$ \eqref{eq:R MMSE DE} amounts to: %\footnote{$\left[{\bf Q}^2\right]_{kk} = \sum_{i=1}^K \left|\left[{\bf Q}\right]_{ki}\right|2$} : 
{\begin{align}
&\overline{\rm SE}_k^{\rm conv,S} =\left(1-\frac{\tau}{T}\right) \log\left(\frac{\rho_d}{\left[{\bf Q}\right]_{kk}}\right)+\mathcal{O}\left(\frac{1}{N}\right). \label{eq:R MMSE DE simplified}
\end{align}
Another key point in $\overline{\rm SE}_k^{\rm conv,S}$ \eqref{eq:R MMSE DE}, is the term ${\sum_{\substack{i=1\\ i\neq k}}^{K}}\left|\frac{1}{\rho_d}\left[{\bf Q}\right]_{ki}\right|^2 $, which represents an approximation of intra-cell interference. If Rayleigh fading is considered (\textit{i.e.} $\overline{\bH}=\mathbf{0}_{N\times K}$), this term cancels out. Pursuant to \cite{HowMany-Jacob2013}, this result confirms that in the setting $N \rightarrow \infty$ while $K$ is fixed, intra-cell interference due to Rayleigh fading dissipates. However, as we can see in this paper, in Rician fading, the specular signals generate intra-cell interference which is embodied by the inner products between $\overline{\bh}_k$ and $\overline{\bh}_i$, $(i,k, 1,\dots,K)$. In light of this outcome, one way to eliminate interference is to have LoS components that are mutually orthogonal between users. This  circumstance can be accomplished under asymptotic favorable propagation conditions where, $\frac{1}{N}\overline{\bh}_i\herm\overline{\bh}_j\asto 0$, $i\neq j$, therefore yielding 
$ {\sum_{\substack{i=1\\ i\neq k}}^{K}}\left|\frac{1}{\rho_d}\left[{\bf Q}\right]_{ki}\right|^2  \asto 0.$ Hence, with the elimination of interference, we can conclude that for Rician fading, better performances are achieved in favorable propagation environments, specifically:  
\vspace{-1em}
\begin{cor}[Favorable Propagation] \label{cor: rate MMSE CS}
if $\frac{1}{N}\overline{\bh}_i\herm\overline{\bh}_j\asto 0$, for $i\neq j$, we have:
\begin{align}
\overline{\rm SE}_k&^{\rm conv,S} =\left(1-\frac{\tau}{T}\right)  %\notag \\ &
 \log\left[1+  \frac{\rho_d}{N}\left(\tr \tilde{\bf R}_k +\Vert\overline{\bh}_k\Vert^2\right)\right]+\mathcal{O}\left(\frac{1}{N}\right).
%\\& = \left(1-\frac{\tau}{T}\right)   \log\left[1+ \frac{\rho_d\beta_k}{1+\kappa_k}\left(\frac{1}{1+\frac{1+\kappa_k}{\rho_d \beta_k \tau}} + \kappa_k\right)\right].
\label{eq:rate MMSE fav cond}
\end{align}
\end{cor}
Furthermore, under the same settings of corollary \ref{cor: rate MMSE CS}, we demonstrate in \cite{Avergae-Ikram2017}, where we compare LMMSE and Matched Filters (MF), that the UL SE \eqref{eq:rate MMSE fav cond} is in fact identical to when MF is used. Accordingly, we find that in massive MIMO with Rician fading channels, LMMSE and MF receivers attain comparable performances only under favorable propagation conditions. This is, however, different from Rayleigh fading, wherein similar performances are obtained by the receivers \textit{(see \cite[Eq.(13)]{Noncooperative-Marzetta2010} and \cite[Remark 3.4]{HowMany-Jacob2013})}.
On another note, from the expression of channel estimates $\wbh_k$ \eqref{eq:h_hat}, it can be shown by a simple eigenvalue decomposition that for low CSI, ($i.e.$ $\tau\rho_{tr} \rightarrow 0$), $\tr \tilde{\bf R}_i \rightarrow 0$, $ \forall \ i$. In such a case, we can see from the SE expressions $\overline{\rm SE}_k^{\rm conv,S}$ \eqref{eq:R MMSE DE simplified} and \eqref{eq:rate MMSE fav cond}, that  the UL performances degrade with the deterioration of  the CSI quality, and ${\rm SE}_k^{\rm conv,S}$ will be mainly determined by the specular signals. Accordingly, we can state that the strength of the LoS component is peculiarly beneficial when the channels are poorly estimated. By the same token, having reliable CSI becomes of greater importance as the LoS component weakens. Consequently, good channel estimates highly impact the performances; however, in Rician fading channels, it is of utmost relevance to have a receiver that exploits the presence and strength of the specular signals in an efficient manner.  
}
\subsubsection{Optimal Training} %\hfill \break
\quad Define $\gamma_k\left(\tau\right)$, such that \eqref{eq:R MMSE DE simplified} writes: $\overline{\rm SE}_k^{\rm conv,S} =$ $\left(1-\frac{\tau}{T}\right)\log\left(1+\gamma_k\left(\tau\right)\right)$.
{ We determine in the next Theorem the optimal value $\tau^*$ that maximizes the achievable average SE.  The objective is to determine the optimal number of symbols out of the total coherence symbols to be dedicated for training, for a fixed power allocation. Therefore, $\tau^*\geq K$ to preserve orthogonality of the pilot sequences, and evidently, $\tau^*< T$}. Accordingly, $\tau^*$ is solution to the optimization problem: \vspace{-1.5em}
 \begin{align*} \tag{P2}\label{P:tau opt}
\tau^{*}=&\underset{\tau} {\rm argmax} \ \frac{1}{K} \sum_{k=1}^{K} \overline{\rm SE}_k^{\rm conv,S}, \\\vspace{-2em}
&{\bf \it s.t.} \ K \leq \tau<T. \notag
\end{align*}
\vspace{-2 em}
\begin{theorem}[Optimal training]\label{th:optim} Under imperfect channel estimates, the optimal training length is given by : 
\begin{itemize}
\item If:
\begin{equation}\label{eq:K tau opti}
\frac{1}{K} \sum_{k=1}^{K} \left[\left(T-{K}\right) \frac{\gamma'_k \left(K\right)}{1+\gamma_k\left(K\right)}-\log\left(1+\gamma_k \left(K\right)\right)\right]\leq 0, 
\end{equation}
 then $\tau^*= K$. 
\item \textit{Otherwise:} $\tau^*$ is the solution to the fixed point equation: 
\begin{equation}\label{eq:tau_opt fix}
\tau^* = T - \frac{\frac{1}{K} \sum_{k=1}^{K}\log\left(1+\gamma_k(\tau)\right)}{\frac{1}{K} \sum_{k=1}^{K}\frac{\gamma'_k (\tau)}{1+\gamma_k(\tau)}},
\end{equation}
where $\gamma_k'\left(\tau\right)$ is the derivative of $\gamma_k\left(\tau\right)$ with respect to $\tau$.
\end{itemize}
\end{theorem}
\vspace{-1em}
\begin{proof}[Proof]
A proof of Theorem \ref{th:optim} is given in Appendix \ref{app:tau_opt}. 
\end{proof}

\par Note that the derivation of $\tau^*$ relies on $\overline{\rm SE}_k^{\rm conv,S}$ considering its explicit form relatively to the intractable alternative, ${\rm SE}_k^{\rm conv,S}$. The results of Theorem \ref{th:optim} will be validated by simulations. Nevertheless, they can be exploited to get some insights on the behavioural tendencies of the choice of $\tau$ and the overall uplink performances. For instance, an interesting direction is the impact of the Rician factor $\kappa$ on the choice of $\tau$. In order to investigate this point, we consider the following case study.
\subsubsection*{\bf Case study}  Let us examine the case where if $K$ is kept fixed and $N$ grows without bound:$\frac{1}{N}\overline{\bh}_i\herm\overline{\bh}_j\asto\frac{\beta_k\kappa_k}{1+\kappa_k}\delta_{ij}$, where $\delta_{ij}$ is the Kronecker  delta. Additionally, let $\forall k: {\bf R}_k = \frac{\beta_k}{1+\kappa_k} \bI_N$, $\beta_k=\beta$ and $\kappa_k=\kappa$. Thus, according to corollary  \ref{cor: rate MMSE CS}: 
\begin{equation}
\gamma\left(\tau\right) = \frac{\beta \rho_d}{1+\kappa}\left(\frac{1}{1+\kappa+\frac{1}{\tau \rho_{tr}}} + \kappa\right).
\end{equation}
%Therefore, % { $\gamma'\left(\tau\right)= \frac{d^2}{\tau^2\left(d+\frac{1}{\tau \rho_d}\right)^2}$  }
\paragraph{Low Rician factor} consider small values of $\kappa$ % $i.e.$ $d\rightarrow\beta$ and $\overline{d}\rightarrow 0$.
\begin{itemize}
\item At a low SNR level ( $\rho_d$ approaches $0$),  the solution \eqref{eq:tau_opt fix} can be rewritten as : 
\begin{equation}
\lim_{\kappa\rightarrow 0} \tau^* =  T- \frac{\left(1+\beta \rho_{tr}\tau\right)\left(1+\beta \rho_{tr}\tau+\beta^2\rho_{tr}\tau\rho_d\right)\log\left(1+\frac{\beta^2\rho_{tr}\tau}{1+\beta \rho_{tr}\tau}\right)}{\beta^2\rho_t\rho_d}.
\end{equation}
Using Taylor's expansion in the low SNR regime yields : 
\begin{equation}\label{eq: tau_opt low snr}
\lim_{\kappa\rightarrow 0} \tau^* = \max \left\{K, \frac{-1+\sqrt{1+\beta\rho_{tr} T}}{\beta \rho_{tr}} \right\}.
\end{equation}
%\item { At high SNR levels ($\rho_d \rightarrow \infty$) : 
%$\tau^*$ is given in \eqref{eq:tau_opt high snr},
%\begin{figure*}
%\begin{equation}\label{eq:tau_opt high snr}
%\tau^* = \max \left\{K, \frac{-1 - (2 + v) \log v+ \sqrt{2 u + 
%  (4 + 2v (u + 1 + v T)) \log v + (2 + v)^2 \log^2 v}}{1 + 2 v + 2 v (1 + v) \log v} \right\}
%\end{equation}\\
%\hrulefill
%\vspace{-0.6cm}
%\end{figure*}
%with $v = \beta \rho_d $ and $u = 1 + 2 v^2 T + v (1 + T)   $}
\end{itemize}
\paragraph{High Rician factor} For high values of the Rician factor, $\kappa$,% $i.e. \ d\rightarrow0$ and $\overline{d}\rightarrow \beta$,
thus leading to $\gamma'(K)=0$, and therefore, \eqref{eq:K tau opti} is always verified, hence : 
 \begin{equation}
\lim_{\kappa \rightarrow \infty} \tau^*=K. 
\label{eq:high LOS K}
\end{equation}
{
This case study sheds some light on how the optimal number of training symbols depends on the large-scale fading parameters, the number of users, the UL SNR and the coherence interval. The first example represents the case wherein the Rayleigh fading is governing at poor SNR levels. As can be seen from \eqref{eq: tau_opt low snr}, $\tau^*$  depends on the system parameters and on the available SNR during training, $\rho_{tr}$. For instance, if this latter is also low, \eqref{eq: tau_opt low snr} yields  $\lim_{\kappa\rightarrow 0} \tau^* = \max \left\{K, \frac{T}{2} \right\}$. This result implies that in a network setting where $T>2K$, to ensure the best performances, half of the total transmitted symbols should be dedicated to training and the other half to useful data. Conversely, if more users are considered such that $T\leq 2K$, then, the optimal number of training symbols should not go beyond the imposed minimum, $K$. 
}
In the second example, as the Rician factor takes higher values, we find that $\tau^*$ always approaches $K$ \eqref{eq:high LOS K}. Consequently, in such circumstances, there is no need to perform any optimization since the optimal number of training symbols is limited to the minimum possible value to ensure pilot orthogonality, namely $K$. More importantly,  we deduce that above a certain $\kappa$, investing in more training samples is not optimal in terms of spectral efficiency as it will have a minor impact on the achievable UL SE. In fact, it might even induce performance losses when $\tau \gg K$, as shall be illustrated in simulations. This result motivates us to analyse, in the next section, the potential outcomes of employing the statistical receiver $\overline{\bg}_k$ \eqref{P1} in LoS-prevailing environments.
%{%\%color{blue}
%
%This example sheds light on the case wherein the Rayleigh fading is governing at low SNR levels  which allows comparison with the above scenario in which the LoS component is dominant \eqref{eq:high LOS K}. Interestingly, the inequality in \eqref{eq:high Ray K} reveals that {\color{red} $\tau^*=K$ is more probable to occur since the condition on $K$ is verified for a considerably smaller number of users than that in \eqref{eq:high LOS K}.}
%%
%%
%For example, at low SNR levels, the condition in \eqref{eq:high LOS K} becomes $K\geq \frac{\sqrt{T}}{\left(1+\beta\rho_d\right)}$. $\tau*$ to be equal to $K$, $\rho_d \geq \frac{1}{\beta}\left(\frac{\sqrt{T}}{K}-1\right)$. 
%This means that if the number of users is small enough compared to $\sqrt{T}$, there exists a $\tau^*$ that maximizes the performances and it is given by : 
% $\tau^* = \floor[\Bigg]{\frac{\sqrt{1+4\left(1+\beta\rho_d\right)^2T}-1 }{2\left(1+\beta\rho_d\right)^2}}$. On the other hand, the bigger the number of users, the more likely for $\tau^*$ to take its minimum value $K$.  
%}
\subsubsection{Statistical Combining in Single-Cell Systems}\label{sec:Average}
As previously mentioned, the objective of having a statistical receiver is to eliminate both training and channel estimation and to exploit the presence of the Rician component efficiently. To this end, the proposed combining vector $\overline{\bg}_k$ is obtained through the maximization of a deterministic approximation of the UL spectral efficiency ${\rm SE}_k^{\rm stat,S}$ \eqref{eq:sum rate stat}, as depicted in \eqref{P1}. Taking into account that $\overline{\bg}_k$ is deterministic itself, a direct application of the convergence of quadratic forms lemma \cite{eigenvaluesoutside-silverstein2009} and continuous mapping Theorem \cite{book-Prob95}, yields:  
\begin{theorem}[Statistical combining in single-cell systems]\label{th:statistical rate} Under Assumption \ref{ass:asymptotic},  ${\rm SE}_k^{\rm stat,S} - \overline{\rm SE}_k^{\rm stat,S}\asto 0$, with:
\begin{equation}\label{eq:R stat DE}
\overline{\rm SE}_k^{\rm stat,S} = %\notag\\ &
 \log\left[1 + 
 \frac{\overline{\bf g}_k\herm \left(\frac{\overline{\bh}_k\overline{\bh}\herm_k}{N}\right)\overline{\bf g}_k}
 {\overline{\bf g}_k\herm\left(\frac{1}{N}\sum_{\substack{i=1}}^{K}  {\bf R}_i+\frac{1}{N}\overline{\bH}_k\overline{\bH}_k\herm+\frac{1}{\rho_d} \bI_N\right)\overline{\bf g}_k}\right], 
\end{equation}
where $\overline{\bH}_k$ is obtained by removing the $k-$th column from the LoS-channels matrix $\overline{\bH}$.
\end{theorem}
Using the expression \eqref{eq:R stat DE}, we can now derive $\overline{\bg}_k$ by solving the SE optimization problem \eqref{P1}. Note that the problem \eqref{P1} is the sum of decoupled positive and increasing functions, therefore, a sufficient condition to solve it is to find $\forall$ $k$, $k=1,\dots,K$, $\overline{\bg}_k$ that satisfies: 
\begin{equation}
\overline{\bg}_k= {\underset{\overline{\bg}_k}{{\rm argmax}}} \ \frac{\overline{\bf g}_k\herm \left(\frac{\overline{\bh}_k\overline{\bh}\herm_k}{N}\right)\overline{\bf g}_k}
{\overline{\bf g}_k\herm\left(\frac{1}{N}\sum_{\substack{i=1}}^{K}  {\bf R}_i+\frac{1}{N}\overline{\bH}_k\overline{\bH}_k\herm+\frac{1}{\rho_d} \bI_N\right)\overline{\bf g}_k}. \ \tag{P1'}
\label{P:optimal g}
\end{equation}
It can be seen that \eqref{P:optimal g} is equivalent to Rayleigh quotient and thus admits the solution: 
\begin{equation}
\overline{\bg}_k= \left(\sum_{\substack{i=1}}^{K}  {\bf R}_i+\overline{\bH}_k\overline{\bH}_k\herm+\frac{N}{\rho_d} \bI_N\right)^{-1}\overline{\bh}_k % \propto {\bf u}_{\rm max}\left({\bf W}_k\right),
\label{eq: g_k opt}
\end{equation} 
 %such that ${\bf u}_{\rm max} \left({\bf  W}_{k}\right)$ is the eigenvector associated with the largest eigenvalue of ${\bf  W}_{k}$, $\lambda_{\rm max}\left({\bf  W}_{k}\right)$, where 
% \begin{equation}
% {\bf W}_k = \left(\sum_{\substack{i=1}}^{K}  {\bf R}_i+\overline{\bH}_k\overline{\bH}_k\herm+\frac{N}{\rho_d} \bI_N\right)^{-1}\overline{\bh}_k\overline{\bh}\herm_k.
% \end{equation}
Consequently, %if $\overline{\bg}_k$ is chosen as in \eqref{eq: g_k opt}, then
 under assumption \ref{ass:asymptotic},
%\begin{equation}
%\overline{\rm SE}_k^{\rm stat}= \left[\log\left(1+ \lambda_{max}\left({\bf W}_k\right)\right) \right].
%\label{eq: R_opt}
%\end{equation}
%Additionally, since the channel correlation matrices have asymptotic bounded spectral norm, $i.e.$  $\frac{1}{N} \Vert {\bf R}_k \Vert_{_2} \asto 0$, we find : 
%\begin{equation}\label{eq:lambda_max}
% \lambda_{\rm max} \left({\bf W}_k\right)-\frac{1}{N} \overline{\bh}\herm_k \left(\frac{1}{N} \overline{\bH}_k\overline{\bH}_k\herm + \frac{1}{\rho_d} \bI_N\right)^{-1} \overline{\bh}_k \asto 0 .
%\end{equation}
%Thus,
\begin{align}
\overline{\rm SE}_k&^{\rm stat,S}=%\notag  \\ &
 \log\left[ 1+\frac{1}{N} \overline{\bh}\herm_k (\frac{1}{N} \overline{\bH}_k\overline{\bH}_k\herm + \frac{1}{\rho_d} \bI_N)^{-1} \overline{\bh}_k \right]+ \mathcal{O}\left(\frac{1}{N}\right). \label{eq:R_stat DE simplified}
\end{align}
Furthermore, under favorable propagation conditions:
\begin{align}
\overline{\rm SE}_k^{\rm stat,S}= 
 \log\left(1+ \frac{\rho_d}{N} \Vert\overline{\bh}_k\Vert^2\right)+ \mathcal{O}\left(\frac{1}{N}\right).
 \label{eq:R_stat fav}
 \end{align}
As can be seen from \eqref{eq:R_stat DE simplified}-\eqref{eq:R_stat fav}, the UL performances generated by $\overline{\bg}_k$ are mainly determined by the level of the specular component. That is, the proposed statistical processing is essentially beneficial in LoS-prevailing environments, thereby requiring a certain level of $\kappa$.
 
 \subsubsection{Comparative Analysis {(Case Study)}}\label{sec: Comparative Analysis}
As concluded above and will be illustrated in simulations, the statistical processing is convenient when the specular component is dominant over the scattered signals. For this reason, we aspire here to find a condition on the Rician factor $\kappa_k$ under which the proposed statistical processing outperforms the conventional processing. We examine a simple network setting where: $\frac{1}{N}\overline{\bh}_i\herm\overline{\bh}_j\asto  \frac{\beta_i \kappa_i}{1+\kappa_i} \delta_{ij}$, where $\delta_{ij}$ is the Kronecker delta. Specifically, we determine $\kappa_k$, $\forall k$, s.t.: 
%{ \begin{align*}\label{P:rates comp}
% \left(1-\frac{\tau}{T}\right) &\frac{1}{K} \sum_{k=1}^{K} \log\left[1+  \frac{\rho_d}{N}\left(\tr \tilde{\bf R}_k +\Vert\overline{\bh}_k\Vert^2\right)\right]  \\ & \leq  \frac{1}{K} \sum_{k=1}^{K}
% \log\left(1+ \frac{\rho_d}{N} \Vert\overline{\bh}_k\Vert^2\right).
% \tag{P3}
%\end{align*}}
{ \begin{align*}\label{P:rates comp}
 \log\left(1+  \frac{\rho_d \beta_k \kappa_k}{1+\kappa_k}\right)  %\\ &
  \geq  \left(1-\frac{\tau}{T}\right) &\log\left[1+ \rho_d \left(\frac{1}{N}\tr \tilde{\bf R}_k + \frac{ \beta_k \kappa_k}{1+\kappa_k}\right)\right]
 \tag{P3}
\end{align*}}
\vspace{-1.5em}
\paragraph*{\bf Result}\label{cor:kappa_up} Taking into account that $\tau \in \left[K,T\right)$, $\frac{1}{N}\overline{\bh}_i\herm\overline{\bh}_j\asto  \frac{beta_i \kappa_i}{1+\kappa_i} \delta_{ij}$ and under Assumption \ref{ass:asymptotic}, it can be shown that the statistical processing outperforms the conventional channel-estimate based processing, if $\kappa_k$ verifies the sufficient condition\footnote{As shown in proof, for mathematical convenience, we consider a higher bound than it is necessary.}:
\begin{equation}\label{ineq:kappa cond2}
\kappa_k \geq \frac{\tr\boldsymbol{\Theta}_k}{N} \frac{T-K}{K}. %\frac{1}{2} \frac{\tr\boldsymbol{\Theta}_k}{N} \left(\frac{T}{K}-1\right) \left(1+\sqrt{1+4\frac{K}{\left(T-K\right)\frac{\tr\boldsymbol{\Theta}_k}{N}}}\right).
\end{equation}
%\end{result}
\begin{proof} A proof is given in Appendix \ref{app:kappa_up}.
\end{proof}
Furthermore, if the correlation matrix $\boldsymbol{\Theta}_k$ follows the widely used one-ring model \cite{correlation-Jakes1994} (introduced in the next section \eqref{eq:theta_k}) or the exponential correlation model\cite{Channel-Loyka2001}, then $\frac{1}{N}\tr\boldsymbol{\Theta}_k=1$, $\forall k$, therefore, \eqref{ineq:kappa cond2} writes :
%\begin{equation}\label{eq:cond}
$ \kappa_k \geq \frac{T-K}{K}. $ %\frac{1}{2} \left(\frac{T}{K}-1\right) \left(1+\sqrt{1+4\frac{K}{T-K}}\right).
%\end{equation}
 These inequalities provide a lower bound on a sufficient Rician factor above which the statistical combing is more profitable than the conventional receiver. This bound is  a function of the systems parameters, including the coherence length and the number of users. Moreover, as can be seen, for a fixed $T$, $\frac{T-K}{K}$ %$\frac{1}{2} \frac{\tr\boldsymbol{\Theta}}{N} \left(\frac{T}{K}-1\right) \left(1+\sqrt{1+4\frac{K}{\left(T-K\right)\frac{\tr\boldsymbol{\Theta}}{N}}}\right)$ 
 is a decreasing function of $K$. As a result, the higher is the number of users, the  smaller is the required $\kappa_k$ to enable the use of the proposed statistical receiver with better UL performance and as such, avoid training along with channel estimation and its associated errors. 
{ \vspace{-1.5em}
\section{Performance Analysis in a Multi-cell scenario}\label{sec:multi}
In this section, we extend our analysis of both conventional and statistical combining schemes to a multi-cell scenario. The objective is to examine the impact of inter-cell interference and pilot contamination on the overall performance of such systems. Therefore, we consider a multi-cell network with $L$ cells having each $K$ single-antenna users communicating with an $N-$antennas BS. In this line, we follow the same notations as in the single-cell scenario, except that we add a triple sub-script indication to differentiate the receiving BS from the cell where the UE is located. For example, $\bh_{j\ell k}$ represents the channel linking the $k-$th UE in cell $\ell$ to BS$_j$. Plus, $\bR_{j\ell k}$ refers to the correlation matrix of channel $\bh_{j\ell k}$, etc. %and system model  as in \cite{Asymptotic-Luca2017}.  
%Assuming Gaussian codebooks, we denote by  $\sqrt{{p}} \bx_j \sim\mathcal{CN}(0,p \bI_K)$  the vector of the transmitted data symbols sent by all the UEs in cell $j$ with the average power ${p}$. 
Accordingly, the received signal at BS$_j$ is given by : %{\color{red} make it $p_i \rho_j N_j K_j$ and so on, mention $h_{j\ell i}$ ici}
\vspace{-0.5em}\begin{equation}
{\bf y}_j= \sqrt{{p}} {\sum_{\substack{\ell=1}}^{L}}\sum_{i=1}^K \bh_{j\ell i} \bx_{\ell i} + {\bf n}_j,
\end{equation}
where ${\bf n}_j$ represents a zero-mean additive Gaussian noise with variance $\sigma^2\bI_N$. Plus, we consider correlated Rician fading for intra-cell or local channels and correlated Rayleigh fading for channels from other cells. { This is a reasonable setting for inter-cell channels, owing to the longer distances between UEs and the BSs in other cells, that would likely include scatterers and thus significantly reduce the possibility of a Line-of-Sight transmission.} Specifically, the channel linking BS$_j$ to UE $k$ located in cell $\ell$ is modeled as:
\begin{align}
\bh_{j\ell k}= & \sqrt{\beta_{j\ell k}} \left(\sqrt{\frac{1}{1+\kappa_{jk}}}\boldsymbol{\Theta}_{j\ell k}^{\frac{1}{2}}\bz_{j\ell k} + \delta_{j \ell} \sqrt{\frac{\kappa_{jk}}{1+\kappa_{jk}}}\overline{\bz}_{jk}\right),
%\bh_{j\ell k}= & \sqrt{\beta_{j\ell k}} \boldsymbol{\Theta}_{j\ell k}^{\frac{1}{2}}\bz_{j\ell k}, \quad \ell \neq j,
 \end{align}
 where $\delta_{j\ell}$ is the Kronecker delta, %\footnote{Kronecker delta $\delta_{j\ell}=1$ iif $\ell=j$, and $\delta_{j\ell}=0$ otherwise.} 
 and $\beta_{j\ell k }$ accounts for the large-scale fading. % Similarly to the single-cell, the channel model includes a Rayleigh-distributed component $\bz_{j\ell k} \sim \mathcal{CN}\left(0,\bI_N\right)$ to include the scattered signals. Plus, for the intra-cell channels, i.e. $\ell=j,$ a deterministic component $\overline{\bz}_{j k}\in\mathbb{C}^{N}$ is considered to represent the specular (LoS) signals. 
 Finally, to simplify the presentation of the results, let %:
 %\begin{align}
${\bf R}_{j \ell k} ={\frac{\beta_{j \ell k} }{1+\kappa_{jk}\delta_{j\ell}}}\boldsymbol{\Theta}_{j \ell k}, $
%\end{align}
 and the aggregate matrix of the LoS components in cell $j$ denoted $\overline{\bH}_{j}=\left[\overline{\bh}_{j1} \overline{\bh}_{j2} \dots \overline{\bh}_{jK}\right]$, with 
%\begin{align}
$\overline{\bh}_{jk} =\sqrt{\frac{\beta_{jjk}\kappa_{jk}}{1+\kappa_{jk}}}\ \overline{\bz}_{jk}.$ \vspace{-1em}
%\end{align}
{%\color{blue}
%\subsection*{Detection and Achievable performances}
\subsection{Conventional combining in Multi-Cell Systems}\vspace{-0.5em}
\par  Similarly to the single-cell scenario, to design the conventional receiver, we consider a pre-training phase of $\tau$ symbols in each cell. On the other hand, to account for pilot contamination, we assume that the same set of pilot sequences is reused in every cell. More specifically, the same pilot is assigned to every $k$-th UE in each cell, and as such $\forall (j,k)$, the estimates of the channels $\bh_{j 1 k}, \bh_{j2k},\dots,\bh_{jLk}$, will be correlated. Accordingly, using the MMSE estimation, the estimate of $\bh_{jj k}$ is given by\cite{book-Kay97}:  \vspace{-0.5em}
 \begin{equation}
{\wbh}_{jj k}= {\bf R}_{jj k } \bPhi_{jk} \left( \sum_{\substack{\ell'=1 }}^{L} \bh_{j \ell' k } + \frac{1}{\sqrt{\tau \rho_{tr}}} {\bf n}_{jk}^{tr}\right)+ \delta_{j \ell} \overline{\bh}_{jk},
\label{eq:h_hat M-MMSE}
\end{equation}
where
$\bPhi_{jk} =\left(\sum_{\ell' =1}^{L}{\bf R}_{j \ell' k} + \frac{1}{\tau \rho_{tr}}\bI_N\right)^{-1}.$ Therefore, ${\wbh}_{jj k}$ $\sim \mathcal{CN}\left(\overline{\bh}_{jk}, \tilde{\bR}_{jj k}\right)$, with $\tilde{\bf R}_{jj k}={\bf R}_{jj k} \bPhi_{j k}{\bf R}_{jj k}$. Plus, the estimation error $\boldsymbol{\xi}_{jj k}=\bh_{jj k}-\wbh_{jj k} $, follows the distribution $\boldsymbol{\xi}_{jj k}\sim \mathcal{CN}\left(0,\bR_{jj k}-\tilde{\bR}_{jj k}\right)$. 
%{\color{red} In MRC and S-MMSE only local channels are estimated and used to process the received signals.}
%%%%\subsection{Detection and Uplink Rates}
Let $\bg_{jk}\in \mathbb{C}^{N\times 1}$ denote the conventional combining vector that BS$_j$ uses to process the signal sent by its UE $k$. This vector is given by \cite{HowMany-Jacob2013}:  
\begin{align}
&{\bg}_{jk} =  \left(\sum_{i=1}^{K}\wbh_{jji}\wbhh_{jji} + \bA_j + \frac{N}{\rho_d} \bI_N \right)^{-1}\wbh_{jjk}.\label{eq:G_MMSE_multi}
\end{align}
 where $\bA_j$ $\in\mathbb{C}^{N\times N}$ is an arbitrary hermitian positive semi-definite design parameter. For instance, it could contain the covariances of estimation errors and inter-cell interference as in \cite{HowMany-Jacob2013}. Therefore, we shall put:  $
\bA_j= \sum_{i=1}^{K}({\bf R}_{jji}-\tilde{\bf R}_{jji})+\sum_{\substack{\ell=1 \\ \ell\neq j}}^{L}\sum_{i=1}^{K}{\bf R}_{j \ell i}$.  Accordingly, the achievable SE corresponding to this transmission, ${\rm SE}_{jk}^{\rm conv,M}$, is defined as:
{\begin{align}\label{eq:gamma def}
&{\rm SE}^{\rm conv,M}_{jk}= %\notag \\ &  
\left(1-\frac{\tau}{T}\right) \mathbb{E} \left[\log\left(1+
\frac{\vert{\bf g}_{jk}\herm\wbh_{jjk}\vert^2}
{\mathbb{E}\left[{\bf g}_{jk}\herm\left({\sum_{\substack{(\ell,i)\neq (j, k)}} \bh_{j\ell i}\bh_{j\ell i}\herm } + \left.\boldsymbol{\xi}_{jk}\boldsymbol{\xi}_{jk}\herm+\frac{1}{\rho_d}\bI_N\right){\bf g}_{jk}\right|\wbH_{j}\right]}\right)\right].
\end{align}}
%Additionally, the achievable UL rate at BS$_j$ is: 
%\begin{align}\label{eq:sum-rate}
%{\rm SE}^{\rm conv}_j=\left(1-\frac{\tau}{T}\right) \frac{1}{K} \sum_{k=1}^{K}\mathbb{E} \left[\log\left(1+\gamma_{jk}\right)\right].
%\end{align}
}
%{\color{red}
%\par In the sequel, we investigate the performances of . Considering cell $j$, the conventional receiver, denoted $\bG_j$ is defined as: 
%
% As to the statistical receiver, we assume that each BS acquires information on its local channels only. Therefore, we shall analyze the performances of the single-cell statistical receiver $\overline{g}_j$ defined in the previous analysis in \ref{P1} and \ref{P:optimal g}.
%}
%\subsection{Impact of Inter-cell Interference and Pilot Contamination}
\begin{theorem}[Conventional combining in multi-cell systems]\label{th:multi-MMSE}
 Under Assumption \ref{ass:asymptotic}, we have : ${\rm SE}^{\rm conv,M}_{jk} - \overline{\rm SE}^{\rm conv,M}_{jk}\asto 0 $, such that:%\eqref{eq:R MMSE DE multi}:
{\vspace{-0.7em}\begin{align}
&\overline{\rm SE}^{\rm conv,M}_{jk} = 
\left(1-\frac{\tau}{T}\right)  \notag \\ & 
\log\Bigg[1+
\frac{\left|1- \frac{1}{\rho_d}\left[\overline{\bQ}_j\right]_{kk}\right|^2}
{
{\sum_{\substack{i=1\\ i\neq k}}^{K}}\left|\frac{1}{\rho_d}\left[\overline{\bQ}_j\right]_{ki}\right|^2
%+ \frac{1}{N^2} {\sum_{\substack{i=1}}^{K}} \overline{\bq} \herm_{j,k}  \overline{\bf T}_{j,i} \overline{\bq}_{j,k}
+ \frac{1}{\rho_d}([\overline{\bQ}_j]_{kk} -\frac{1}{\rho_d}\left[\overline{\bQ}_j^2\right]_{kk})
+{\sum_{\substack{\ell=1 \\ \ell \neq j}}^{L}} {\sum_{\substack{i=1}}^{K}}\left|\frac{1}{N} [\overline{\bQ}_j]_{ki}\tr(\bR_{j\ell i}\bPhi_{ji}{\bR}_{jj i})\right|^2 
}\Bigg]. \vspace{-2em} \label{eq:R MMSE DE multi}
\end{align}}
\vspace{-0.5em} with
$\overline{\bQ}_j = \left(\frac{1}{N}\overline{\bH}_j\herm\overline{\bH}_j+ diag\left\{ \frac{1}{N}\tr(\tilde{\bR}_{jji}
)\right\}_{i=1}^{K} +\frac{1}{\rho_d} \bI_K\right)^{-1}$, %\label{eq:Q_j} %\\
and $\overline{\bq}_{j,k}$ is the $k-th$ column of matrix $\overline{\bQ}_j$.  %Furthermore, using the continuous mapping theorem \cite{book-Prob95}, the  uplink rate $R_j$ \eqref{eq:sum-rate} in cell $j$ verifies:  
%\begin{equation}
%R_j - \left(1- \frac{\tau}{T}\right) \frac{1}{K}\sum_{k=1}^K \log\left(1+\overline{\gamma}_{jk}\right) \asto 0. \label{eq:R multi conv}
%\end{equation}
\end{theorem}\vspace{-0.5em}
\begin{proof}
The proof is given in Appendix \ref{app:MMSE multicell proof}
\end{proof}
As can be seen from the SE approximation $\overline{\rm SE}^{\rm conv,M}_{jk}$\eqref{eq:R MMSE DE multi}, the achievable UL SE in the multi-cell setting has an analogous expression to the single-cell scenario \eqref{eq:R MMSE DE}, apart from the inter-cell interference represented by the last term of the denominator of $\overline{\rm SE}^{\rm conv,M}_{jk}$\eqref{eq:R MMSE DE multi}. As a result, most conclusions provided in Theorem \ref{th: MMSE hat} {\textit{(conventional combining in multi-cell systems)}} hold true in the multi-cell setup, including the cancellation of the estimation errors as $N$ grows large without bound. In addition, intra-cell interference is also generated by the inner products between the LoS components, and as such, dissipates under favorable propagation conditions. 
 {\par We now move on to investigating the effect of inter-cell interference on this combining approach.} In this line, by expanding the inter-cell interference approximation, $\overline{\rm SE}^{\rm conv,M}_{jk}$ writes: 
{%\small
%\begin{align}
%&\overline{\rm SE}^{\rm conv}_{jk} = 
%\left(1-\frac{\tau}{T}\right) \frac{1}{K} \sum_{k=1}^{K} \notag \\ & 
%\log\Bigg(1+ 
%\frac{\left|1- \frac{1}{\rho_d}[\overline{\bQ}_j]_{kk}\right|^2}
%{
%\frac{[\overline{\bQ}_j]_{kk}}{\rho_d}-\frac{[\overline{\bQ}_j]_{kk}^2}{\rho_d^2} 
%+\displaystyle{\sum_{\substack{\ell \neq j}}}\Big(\underbrace{\vert\frac{1}{N}\left[\overline{\bQ}_j\right]_{kk}\tr(\bR_{j\ell k}\bPhi_{jk}{\bR}_{jj k})\vert^2}_{\text{induced by pilot contamination}}
%+\underbrace{\displaystyle{\sum_{\substack{i\neq k}}}\vert\frac{1}{N}[\overline{\bQ}_j]_{ki}\tr(\bR_{j\ell i}\bPhi_{ji}{\bR}_{jj i})\vert^2}_{\text{uncorrelated interference}}
%\Big) 
%}\Bigg). \label{eq:R MMSE DE multi simplified}
%\end{align}
\begin{align}
&\overline{\rm SE}^{\rm conv,M}_{jk} = 
\left(1-\frac{\tau}{T}\right)  \notag \\ & 
\log\Bigg(1+ 
\frac{\displaystyle{\Big|\frac{\rho_d}{[\overline{\bQ}_j]_{kk}}-1\Big|^2}}
{
\displaystyle{\frac{\rho_d}{[\overline{\bQ}_j]_{kk}}-1 }
+\displaystyle{\sum_{\substack{\ell \neq j}}}\Big(\underbrace{\vert\frac{\rho_d}{N}\tr(\bR_{j\ell k}\bPhi_{jk}{\bR}_{jj k})\vert^2}_{\text{induced by pilot contamination}}
+\underbrace{\displaystyle{\sum_{\substack{i\neq k}}}\vert\frac{\rho_d}{N}\frac{[\overline{\bQ}_j]_{ki}}{[\overline{\bQ}_j]_{kk}}\tr(\bR_{j\ell i}\bPhi_{ji}{\bR}_{jj i})\vert^2}_{\text{uncorrelated interference}}
\Big) 
}\Bigg).\vspace{-2em} \label{eq:R MMSE DE multi simplified}
\end{align}
}
{Expression \eqref{eq:R MMSE DE multi simplified} separates the pilot contamination induced interference from the remaining inter-cell interference, which we refer to as ``uncorrelated" interference. Clearly, inter-cell interference limits the overall performances even at the infinite-antenna limit. Nevertheless, note that its impact can be alleviated through the mitigation of the uncorrelated interference by observing that this latter is eliminated when $\overline{\bQ}_j$ becomes diagonal. In fact, this is achieved in favorable propagation conditions, wherein the performances will attain:}
\begin{cor}[Favorable propagation in multi-cell]\label{cor:multi fav}
if $\frac{1}{N}\overline{\bh}_{ji}\herm\overline{\bh}_{jk}\asto 0$, for $i\neq k$, we have:
\begin{align}
&\overline{\rm SE}^{\rm conv,M}_{jk} = %\notag \\ & 
\left(1-\frac{\tau}{T}\right) 
\log\Bigg(1+ 
\frac{\left(\frac{\rho_d}{N}\tr \tilde{\bf R}_{jjk} +\frac{\rho_d}{N}\Vert\overline{\bh}_{jk}\Vert^2\right)^2}
{
\frac{\rho_d}{N}\tr \tilde{\bf R}_{jjk} +\frac{\rho_d}{N}\Vert\overline{\bh}_{jk}\Vert^2
+\displaystyle{\sum_{\substack{\ell \neq j}}^{L}}\underbrace{\left(\frac{\rho_d}{N}\tr(\bR_{j\ell k}\bPhi_{jk}{\bR}_{jj k})\right)^2}_{\text{induced by pilot contamination}}
}\Bigg).\vspace{-1em} \label{eq:R MMSE DE multi fav}
\end{align}
%
%\begin{align}
%\overline{\rm SE}&^{\rm conv} =\left(1-\frac{\tau}{T}\right)  %\notag \\ &
%\frac{1}{K} \sum_{k=1}^{K} \log\left[1+  \frac{\rho_d}{N}\left(\tr \tilde{\bf R}_k +\Vert\overline{\bh}_k\Vert^2\right)\right]+\mathcal{O}\left(\frac{1}{N}\right).
%%\\& = \left(1-\frac{\tau}{T}\right)  \frac{1}{K} \sum_{k=1}^{K} \log\left[1+ \frac{\rho_d\beta_k}{1+\kappa_k}\left(\frac{1}{1+\frac{1+\kappa_k}{\rho_d \beta_k \tau}} + \kappa_k\right)\right].
%\label{eq:rate MMSE fav cond}
%\end{align}
\end{cor}
{Another important outcome from Theorem \ref{th:multi-MMSE} and corollary \ref{cor:multi fav} lies in the interplay between the interference emanating from pilot-contamination and the LoS signals. % strength which is epitomized by the value of $\kappa_{jk}$. 
In fact, consider the quantity $\frac{1}{N}\tr({\sum_{\ell\neq j}^L}\bR_{j\ell k}\bPhi_{jk}{\bR}_{jj k})$ in \eqref{eq:R MMSE DE multi simplified} and \eqref{eq:R MMSE DE multi fav} which represents this type of correlated interference. As shown in the following proof, this term is a decreasing function of the Rician factor $\kappa_{jk}$. Actually, in the limiting case  $\kappa_{jk}\rightarrow \infty$, we have $\frac{1}{N}\tr({\sum_{\ell\neq j}^L}\bR_{j\ell k}\bPhi_{jk}{\bR}_{jj k})\rightarrow 0$. Consequently, we can state that in such multi-cell systems, another advantage of having stronger LoS components is to reduce the adverse effects of pilot contamination which is known to be a limiting performance factor in massive MIMO systems \cite{Noncooperative-Marzetta2010,HowMany-Jacob2013}. 
\begin{proof} The proof is given in Appendix \ref{app: pilot vs kappa}
\end{proof}\vspace{-1em}
}
\subsection{Statistical Combining in Multi-Cell Systems}\vspace{-0.5em}
In this section, we propose to investigate whether the statistical receiver defined in the single-cell scenario can still be beneficial in a multi-cell system with LoS-prevailing transmissions. In other words, we are interested in investigating the resilience of this combining scheme when it is subject to inter-cell interference. In this line, let $\overline{\bg}_{jk}$ indicate the statistical combining vector associated with the communication between UE $k$ and its BS $j$,  defined as:\begin{equation}
\overline{\bg}_{jk}= \left(\sum_{\substack{i=1}}^{K}  {\bf R}_{jji}+\overline{\bH}_{j,/k}\overline{\bH}_{j,/k}\herm+\frac{N}{\rho_d} \bI_N\right)^{-1}\overline{\bh}_{jk}, % \propto {\bf u}_{\rm max}\left({\bf W}_k\right),
\label{eq:multi g_k stat}
\end{equation} 
where $\overline{\bH}_{j,/k}$ is matrix  $\overline{\bH}_{j}$ without the $k-$th column. 
Furthermore, since utilizing this receiver allows to circumvent training and estimation, the corresponding UL SE attains: 
\begin{equation}\label{eq:multi rate stat}
{\rm SE}_{jk}^{\rm stat, M}= \mathbb{E}\Bigg[\log\Bigg(1+ 
\frac{\vert\overline{\bf g}_{jk}\herm\overline{\bh}_{jk}\vert^2}
{\mathbb{E}\left[\overline{\bf g}_{jk}\herm\bigg({\sum_{\ell=1}^{K}\sum_{\substack{i=1}}^{K}}{\bh}_{j\ell i}{\bh}\herm_{j\ell i}-\overline{\bh}_{jk}\overline{\bh}_{jk}\herm+\frac{N}{\rho_d}\bI_N\bigg)\overline{\bf g}_{jk}\herm\right]}\Bigg) \Bigg].
\end{equation}
Next, we provide an asymptotic approximation of the achievable SE generated by $\overline{\bg}_{jk}$. This constitutes the last main result of this work. \vspace{-0.5em}
\begin{theorem} [Statistical combining in multi-cell systems]\label{th:multi statistical rate} Under Assumption \ref{ass:asymptotic},  ${\rm SE}_{jk}^{\rm stat,M} - \overline{\rm SE}_{jk}^{\rm stat,M}\asto 0$, with \vspace{-1em}
\begin{equation}\label{eq:multi R stat DE simplified}
\overline{\rm SE}_{jk}^{\rm stat,M} =  \log\Bigg[1+  \overline{\bh}_{jk}\herm \left(\overline{\bH}_{j,/k}\overline{\bH}_{j,/k}\herm+\frac{N}{\rho_d} \bI_N\right)^{-1}\overline{\bh}_{jk} \Bigg],
\end{equation}
\end{theorem}
\begin{proof} Under assumption \ref{ass:asymptotic}, since $\overline{\bg}_{jk}$ is deterministic, a direct application of the convergence of quadratic forms lemma \cite{eigenvaluesoutside-silverstein2009} and the continuous mapping Theorem \cite{book-Prob95} yields $\overline{\rm SE}_{jk}^{\rm stat,M}$.
%The approximation $\overline{\rm SE}^{\rm stat}_{j}$ \eqref{eq:multi R stat DE simplified} stems from the same arguments as in single-cell given in Theorem \ref{th:statistical rate}, and is therefore the details are omitted.
% More specifically, considering that all correlation matrices have bounded spectral norm as $N \rightarrow \infty$ in addition to the fact that $\overline{\bg}_{jk}$ is deterministic, a straightforward application of the convergence of quadratic forms lemma \cite{eigenvaluesoutside-silverstein2009} followed by the continuous mapping Theorem \cite{book-Prob95} yields the asymptotic result. 
%\begin{align}\label{eq:R stat DE}
%&\overline{\rm SE}_j^{\rm stat} = %\notag\\ &
% \log\Bigg(1 + 
% \frac{\overline{\bf g}_k\herm \left( \frac{\overline{\bh}_k\overline{\bh}\herm_k}{N}\right)\overline{\bf g}_k}
% {\overline{\bf g}_k\herm\big(\displaystyle{\frac{1}{N}\sum_{\ell=1}^L\sum_{\substack{i=1}}^{K} } {\bf R}_{j\ell i}+\frac{1}{N}\overline{\bH}_k\overline{\bH}_k\herm+\frac{1}{\rho_d} \bI_N \big) \overline{\bf g}_k}\Bigg), 
%\end{align}
\end{proof}
First of all, we emphasize once again that the receiver, $\overline{\bg}_{jk}$, is purposely designed for environments with strong LoS components. Second, in such environments, comparing the above multi-cell SE $\overline{\rm SE}^{\rm stat, M}_{jk}$ \eqref{eq:multi R stat DE simplified} with the single-cell one $\overline{\rm SE}^{\rm stat, S}_k$ \eqref{eq:R_stat DE simplified}, reveals that employing the statistical combining scheme entails a similar asymptotic performance gain for both network settings, $i.e.$, $\overline{\rm SE}_{jk}^{\rm stat, M}- \overline{\rm SE}^{\rm stat, S}_k\asto 0$.
 Therefore, we can conclude that, for LoS-prevailing communications, inter-cell interference can be mitigated through the use of $\overline{\bg}_{jk}$ and as such, does not constitute a limitation of the achievable capacity when this processing approach is employed. This outcome is explained by the fact that the underlying premise behind the statistical receiver is to bypass training and thereby, prevent pilot-contamination and its ensuing undesirable effects. It is also important to note that this desirable feature comes in contrast to conventional combining, whose UL SE remains limited by pilot-contamination-induced interference, as previously shown in expressions \eqref{eq:R MMSE DE multi simplified}-\eqref{eq:R MMSE DE multi fav}, $i.e.$, $\overline{\rm SE}^{\rm conv, M}_{jk} \leq  \overline{\rm SE}^{\rm conv, S}_k$.  %{\color{red} Finally, since since we proved in section \ref{sec: Comparative Analysis} the existence of $\overline{\kappa}$ above which  $\overline{\rm SE}^{\rm stat,S}_k \geq \overline{\rm SE}^{\rm conv,S}$, in addition to the fact that $\overline{\rm SE}^{\rm conv,S}\geq \overline{\rm SE}^{\rm conv,M}_{jk}$ (due to inter-cell interference), by transitivity $\exists$  $\overline{\kappa}^{\rm \tiny M}$, above which $\overline{\rm SE}_j^{\rm stat,M}\geq \overline{\rm SE}^{\rm conv,M}_{jk}$.}
%inter-cell interference is mitigated when $\overline{\bg}_{jk}$ is used, 
\vspace{-1em}
\section{Numerical Results}\label{sec:Num Results}
In this section, we carry out MonteCarlo simulations over $1000$ channel realizations to validate, for finite system dimensions, the asymptotic results for both single-cell and mutli-cell settings, given in Sections \ref{sec:single} and \ref{sec:multi}.\vspace{-1em}
\subsection{Single-cell scenario}
For this scenario, we consider a single-cell massive MIMO having one BS with $N=150$ antennas, $K=20$ users, and a coherence length $T_c=500$ symbols. The inner cell-radius is $150$m and the users are uniformly distributed around the BS at an arrival angle $\theta_k$. Furthermore, the pathloss is given $ \beta_k =\frac{1}{x_k^\alpha}$, where $x_k$ is the distance between UE $k$ and the BS and $\alpha=2.5$. The specular component $\overline{\bz}_k$ follows the model $\left[\overline{\bf z}_k\right]_ n= e^{-j(n-1)\pi\sin\left(\theta_k\right)}$.  Moreover, to ensure distinct Rician factors among the users, we assume throughout all the simulations that $\forall k$, $ \kappa_k\sim \mathcal{U}\left[0,\kappa_{\rm max}\right]$. %\sim \mathcal{U}(0,\kappa_{\rm max})$. 
As a result, varying $\kappa_{\max}$ yields specular signals with different levels of strength which ultimately enables to epitomize both NLoS and LoS prevailing environments. Finally, the elements of the correlation matrix $\boldsymbol{\Theta}_k$ of channel $\bh_k$ are given by the one ring model \cite{correlation-Jakes1994}:
\begin{equation}\label{eq:theta_k}
\left[\boldsymbol{\Theta}_k\right]_{uv}=\frac{1}{\theta_{k,max}-\theta_{k,min}}  \int_{\theta_{k,min}}^{\theta_{k,max}} e^{j \frac{2\pi}{\lambda} a_{uv} \cos\left(\theta\right)} \mathrm{d}\theta,
\end{equation}
where $\lambda$ denotes the signal's wavelength and $a_{uv}$ is the distance between receive antennas $u$ and $v$. We also choose $\frac{a_{uv}}{\lambda} =0.5\vert u-v\vert$, $\theta_{k,min}=-\pi$ and $\theta_{k,max}=\theta_k-\pi$.
\begin{figure}[h]
%\begin{minipage}[t]{1\textwidth}
\hspace{-0.7em}
\centering
\begin{subfigure}[b]{0.48\linewidth}
\includegraphics[width=\linewidth, height=\height]{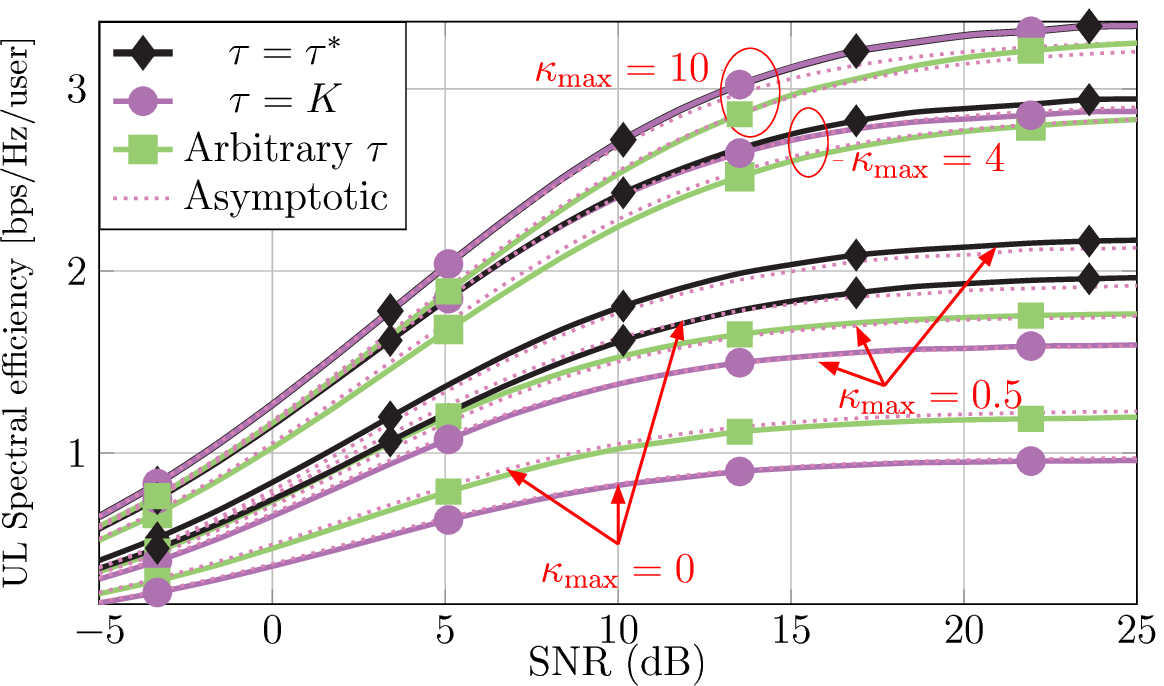}
\vspace{-2em}
\caption{}
\label{fig:Conven_diff_tau_kappa}
\end{subfigure}
%%%%%%%%%%%%%%%%%
%\hspace{0.5em}
%%%%%%%%%%%%%%%%%
\begin{subfigure}[b]{0.48\linewidth}
\includegraphics[width=\linewidth, height=\height]{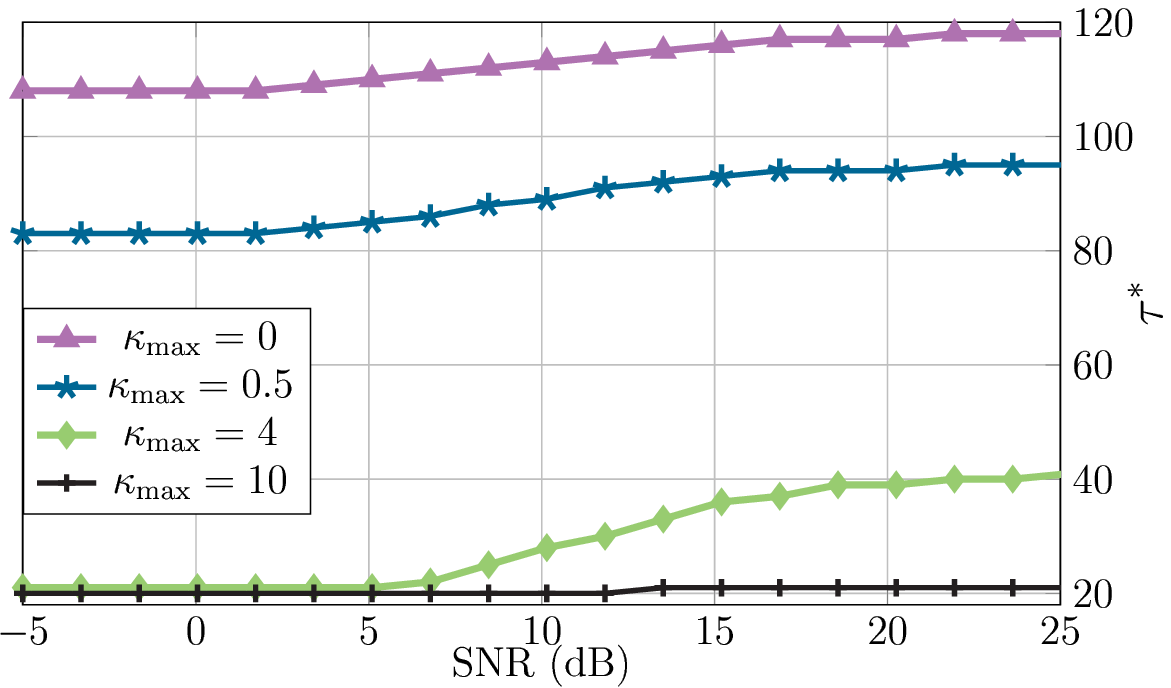}
\vspace{-2em}
\caption{}
\label{fig:tau_opt_vs_kappaz}
\end{subfigure} 
\caption{\small Single-cell setting: (a) Impact of $\tau$ on the UL SE using conventional combining \eqref{eq:G_MMSE_hat}, for different levels of Rician factor, s.t. $\kappa_{k}\sim \mathcal{U}(0,\kappa_{\rm max})$. %Solid and dashed lines represent empirical and asymptotic results, respectively. 
(b) Optimal number of training symbols $\tau^*$ \eqref{eq:tau_opt fix} for different levels of $\kappa_k$.}\label{fig:Conv_SNR}
\end{figure}

\par We first illustrate the effects of the LoS presence and the length of the training sequence on the performances when using the conventional combining to, ultimately, validate the conclusions of Theorems \ref{th: MMSE hat} and \ref{th:optim}. To this end, we plot in Fig.\ref{fig:Conven_diff_tau_kappa} the achievable SE ${\rm SE}^{\rm conv,S}_k$ \eqref{eq:sum rate hat} for different levels of the Rician factors, including, $\kappa_{\rm max}= 0$ (corresponding to Rayleigh fading), $0.5$, $4$ and $10$.  Moreover, to manifest the importance of the number of training symbols, we represent in the same figure ${\rm SE}^{\rm conv,S}_k$ with different values of $\tau$, namely: the minimum $K$, the optimal $\tau^*$ \eqref{eq:tau_opt fix}, and another arbitrary value $(\neq \tau^*$ nor $K)$. Solid and dotted lines represent empirical and asymptotic SEs, respectively. In addition, for each scenario, we plot in Fig.\ref{fig:tau_opt_vs_kappaz} the obtained optimal values $\tau^*$ \eqref{eq:K tau opti}-\eqref{eq:tau_opt fix} with respect to the SNR for the various levels of $\kappa_{\rm max}$.

\par Overall, as expected, the LoS has a beneficial impact since increasing the Rician factor enables higher SEs, for any value of $\tau$. As to this latter, it can be seen that the best performances are clearly obtained when the optimal number of symbols given in Theorem \ref{th:optim} ($i.e.$ $\tau=\tau^*$) is considered (\textit{Fig.\ref{fig:Conven_diff_tau_kappa}, diamond-marked curves}). Furthermore, note that the gap between the settings $\tau=\tau^*$ and $\tau=K$ is particularly noteworthy at small values of Rician factors, (\textit{Fig.\ref{fig:Conven_diff_tau_kappa}, $\kappa_{\rm max}= 0 $ and $0.5$}, $i.e$ $0\leq\kappa_k\leq 0.5$, $\forall k$). However, this difference in performance becomes less significant as $\kappa_k$ takes higher values (\textit{Fig.\ref{fig:Conven_diff_tau_kappa}, curves $\kappa_{\rm max}=4$ and $10$, $i.e.$ $0\leq \kappa_k \leq 10$}). These simulation results confirm that as $\kappa$ takes higher values, $\tau^*\rightarrow K$ as also displayed in Fig.\ref{fig:tau_opt_vs_kappaz}. Indeed, Fig.\ref{fig:tau_opt_vs_kappaz} clearly asserts that, for small Rician factors ($\kappa_{max}=0$ and $0.5$), $\tau^*$ takes increasingly higher levels with the increase of the SNR. 	Conversly,  $\tau^*\rightarrow K$ when $\kappa_{\rm max}=$ $10$. Evidently, since $\tau^*\rightarrow K$ for these scenarios, the SEs corresponding to $\tau=K$ and $\tau=\tau^*$ are almost identical (\textit{Fig.\ref{fig:Conven_diff_tau_kappa}, overlapping diamond and circle-marked curves}); whereas, interestingly, for $\tau > K$ (represented by the square-marked curve), we observe a decrease in the SE. Consequently, the plots in Fig.\ref{fig:Conven_diff_tau_kappa} and Fig.\ref{fig:tau_opt_vs_kappaz} validate the conclusions of Section \ref{sec:Analysis} and the case study indicating the importance of assigning the optimal number of training symbols to attain the best performance. Additionally, as another important result, these simulations manifest that above a certain level of $\kappa_{\rm max}$, $i.e.$ as the LoS get stronger, investing in longer training sequences to enhance the spectral efficiency is actually counterproductive. This interesting outcome inspired the proposed statistical combining that is a more opportune approach in such environments, as was previously demonstrated in Section \ref{sec:single}, and is illustrated in what follows.
 %%%%%%%%%%%%%%%%%%%%%%%%%%%%%%%%%
 %%%%%%%%%%%%
\begin{figure} %%%%%%%%%%%%%%%%%%%%%%%% Single-Cell Fig2 CONV STAT
\begin{subfigure}[b]{0.48\linewidth}
\includegraphics[width=1\linewidth, height=\height]{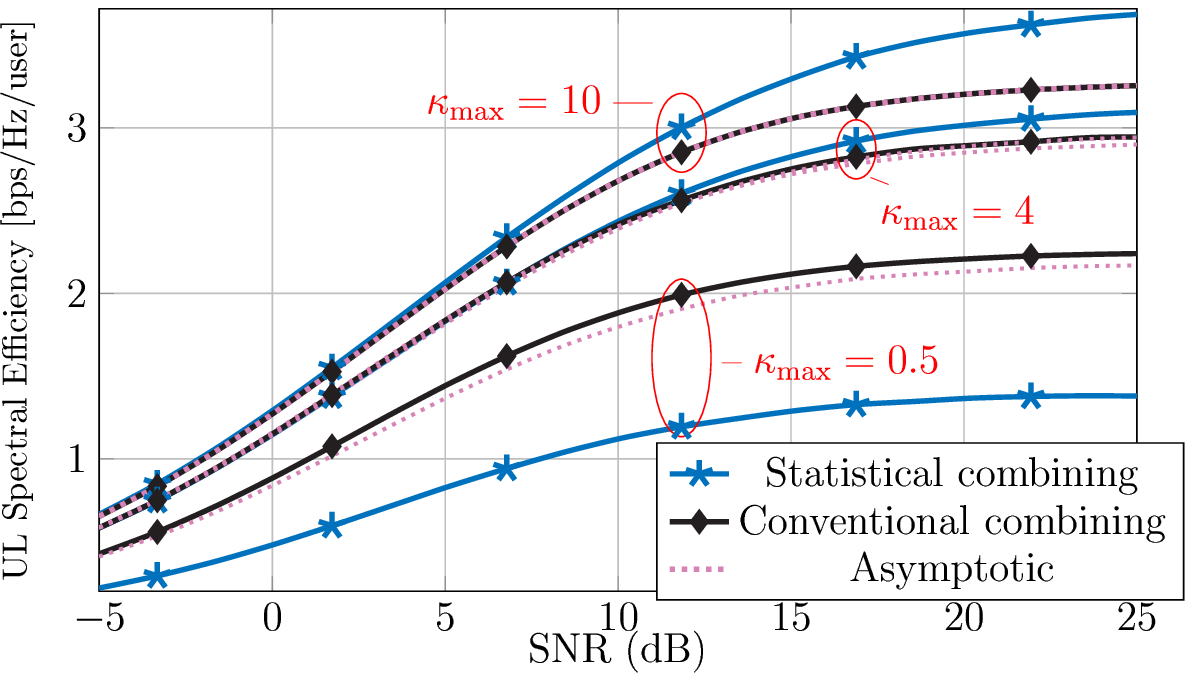}
\vspace{-2em}
\caption{\small Ordinary propagation}
\label{fig:Conv_Statistical_comparison_Ordin}
\end{subfigure}
%\hspace{0.1 em}
%%%%%%%%%%%%%%%
\begin{subfigure}[b]{0.48\linewidth}
\includegraphics[width=1\linewidth, height=\height]{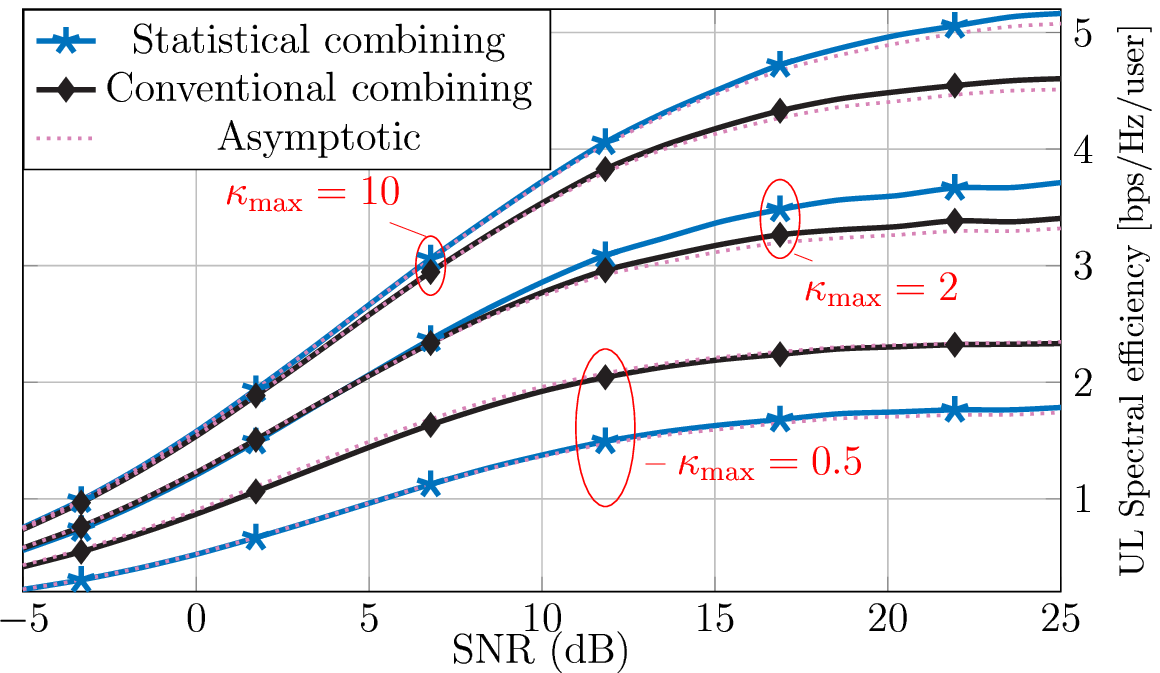}
\vspace{-2em}
\caption{\small Favorable propagation}
\label{fig:Conv_Stat_comp_Fav}
\end{subfigure}
%%%%%%%%%%
\caption{\small Single-cell setting: UL SE using conventional combining \eqref{eq:G_MMSE_hat} with optimal training $\tau^*$ and statistical combining \eqref{eq: g_k opt} for different levels of $\kappa$, with (a) ordinary  and (b) favorable  propagation conditions.}
\vspace{-2em}
\label{fig:Conv_Stat}
\end{figure}
%%%%%%%%%%%%%%%%%%%%%%%%%555
\par Under the aforementioned network setting and for different values of $\kappa$, we compare in Fig.\ref{fig:Conv_Stat} the UL SE ${\rm SE}^{\rm stat,S}_k$ \eqref{eq:sum rate stat} achieved using the statistical receiver with the one attained by the conventional technique ${\rm SE}^{\rm conv,S}_k$ \eqref{eq:sum rate hat} assuming optimal training, $i.e.$ $\tau=\tau^*$. Plus, we represent ordinary and favorable propagation conditions in Fig.\ref{fig:Conv_Statistical_comparison_Ordin} and Fig.\ref{fig:Conv_Stat_comp_Fav}, respectively. First, comparing Fig.\ref{fig:Conv_Statistical_comparison_Ordin} with Fig.\ref{fig:Conv_Stat_comp_Fav} reveals that favorable propagation enable better performances for both combining techniques. This consequence, as explained in Theorem \ref{th: MMSE hat}, is due to the cancellation of LoS induced intra-cell interference when the specular signals are mutually orthogonal. Second, as can be seen in both propagation conditions, conventional LMMSE is more beneficial than the statistical combiner at low ranges of the Rician factor, (\textit{Fig.\ref{fig:Conv_Stat}, $\kappa_{\rm max}=0.5$}). Nonetheless, as the LoS component becomes stronger, ${\rm SE}^{\rm stat,S}_k$ progressively approaches ${\rm SE}^{\rm conv,S}_k$, up to generating exceeding gains starting at $\kappa_{\rm max}=4$ in ordinary conditions, and $\kappa_{\rm max}=1.5$ in favorable propagation. {}This consequence can be justified by the expression $\overline{\rm SE}^{\rm stat,S}_k$ that clearly demonstrates that the statistical receiver's performance is mainly determined by the strength of the LoS components. Therefore, these results confirm our single-cell analysis by substantiating  the existence of a $\overline{\kappa}$ above which the statistical processing outperforms the conventional one. In the same line, this threshold value is fairly lower in favorable propagation compared to ordinary propagation environments. It is also important to note that  they extend the outcome in Section \ref{sec: Comparative Analysis} to a more realistic scenario that accounts for different per-user correlations and Rician factors. Finally, Fig.\ref{fig:Conv_SNR} and Fig.\ref{fig:Conv_Stat} validate, for finite system dimensions, the accuracy of the asymptotic approximations derived in Theorems \ref{th: MMSE hat}, \ref{th:optim} and \ref{th:statistical rate}. \vspace{-1em}
\subsection{Multi-cell scenario}
For the multi-cell scenario, we consider $L=3$ adjacent cells having the same parameters as defined in the single-cell section. Besides, for each cell, we consider cell-edge users as shown in Fig.\ref{fig:Netset}. Deploying the users in such a configuration generates high levels of inter-cell interference, and the close angles of arrival ensures considerable intra-LoS interference.
%%%%%%%%%%%%%%%%%%%%%%%%%% Multicell NetSetup FIg3
\hspace{-1.3em}\vspace{-1em}
\begin{figure}[h]
\scalebox{0.9}{
\begin{minipage}[c]{0.6\linewidth}
\centering
\includegraphics[width=1\linewidth, height=\height]{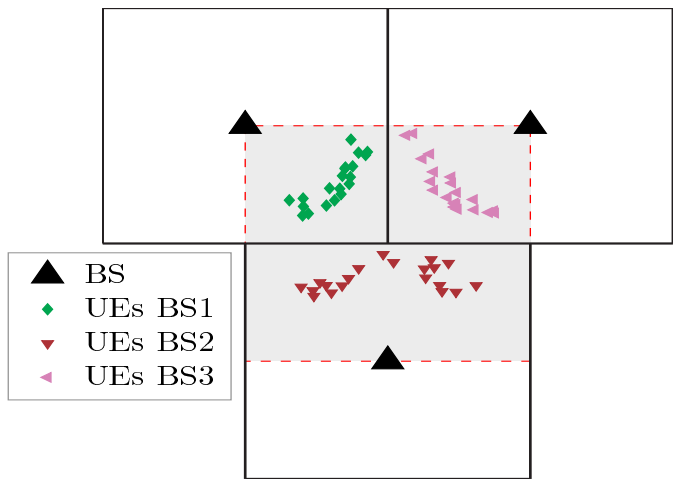}
\end{minipage}\hfill \hspace{0.5em}
\begin{minipage}[c]{0.4\textwidth}
    \caption{ Multi-cell network setup with L=3 cells and K=20 cell-edge users.}\label{fig:Netset}
\end{minipage}}
\end{figure}
\vspace{-1em}
%%%%%%%%%%%%%%%%%%%%%%%%%%%
\par We illustrate in Fig.\ref{fig:Multi_Conv_Stat} the achievable SEs for both statistical and conventional combining schemes considering pilot contamination and correlated Rician fading.  Fig.\ref{fig:Multi_conv_Stat_comp_ordin} and Fig.\ref{fig:Multi_conv_Stat_comp_fav} account for ordinary and favorable propagation conditions, respectively. In accordance with the discussion in section \ref{sec:multi}, these figures confirm that the multi-cell UL spectral efficiencies follow the same pattern as those observed for a single-cell system. That is, firstly, higher Rician factors entail increasingly better performances. Secondly, favorable propagation conditions further enhance the SE, due to the cancellation of the uncorrelated inter-cell interference for the conventional combining, and the intra-LoS interference for both receivers, as analytically demonstrated in Theorems \ref{th:multi-MMSE} and \ref{th:multi statistical rate}. \vspace{-1em}
%%%%%%%%%%%%%%%%%%%%%%%%%% MutltiCell Fig4.a
\begin{figure}
\begin{subfigure}[b]{0.47\linewidth}
\includegraphics[width=1\linewidth, height=\height]{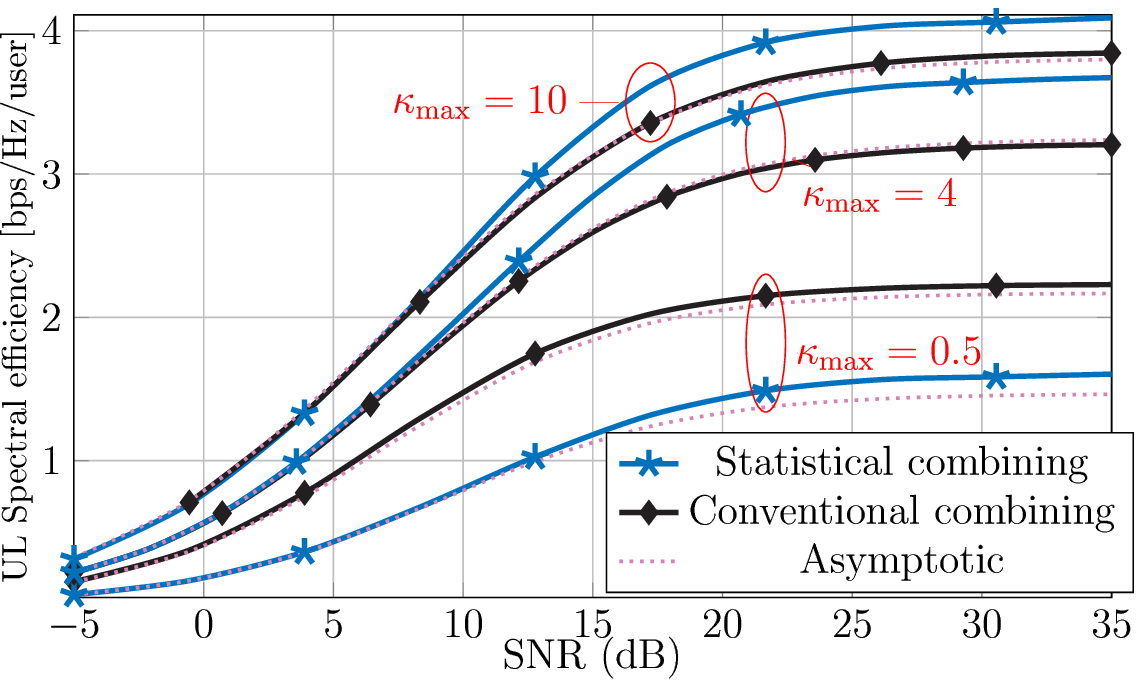}
\vspace{-2em}
\caption{\small Ordinary propagation}
\label{fig:Multi_conv_Stat_comp_ordin}
\end{subfigure}
\hspace{0.9em}
%%%%%%%%%%%%%%%%%%%%%%%%%%MutltiCell Fig4.b
\begin{subfigure}[b]{0.47\linewidth}
\includegraphics[width=1\linewidth, height=\height]{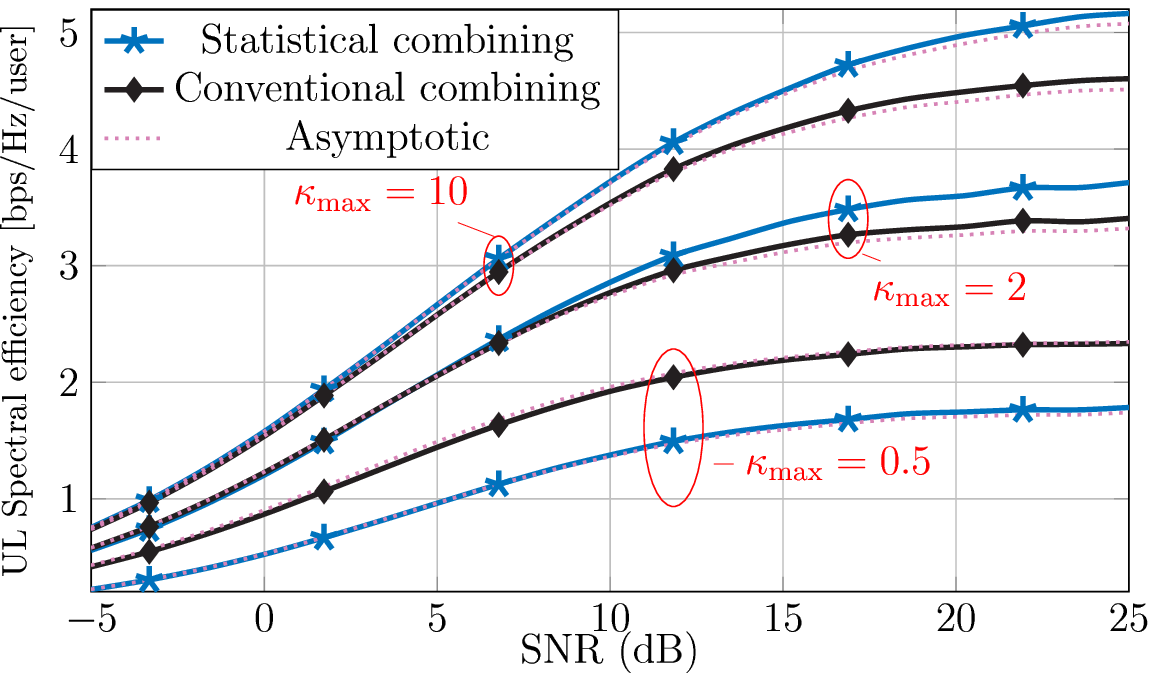}
\vspace{-2em}
\caption{\small Favorable propagation}
\label{fig:Multi_conv_Stat_comp_fav}
\end{subfigure}
\caption{\small  Multi-Cell setting: UL SE using multi-cell conventional combining \eqref{eq:G_MMSE_multi} with $\tau=K$ and statistical processing \eqref{eq:multi g_k stat} with different levels of Rician factor, in (a) ordinary and (b) favorable propagation conditions.
}
\vspace{-1.5em}
\label{fig:Multi_Conv_Stat}
\end{figure}
%%%%%%%%%%%%%%%%%%%%%%%
%%%%%%%%%%%%%%%%%%%%%%%%%%%%%%%%% MultiCell_SingleCell Comparison
\begin{figure}[h]
\begin{minipage}[c]{0.6\linewidth}
\includegraphics[width=0.9\linewidth, height=\height]{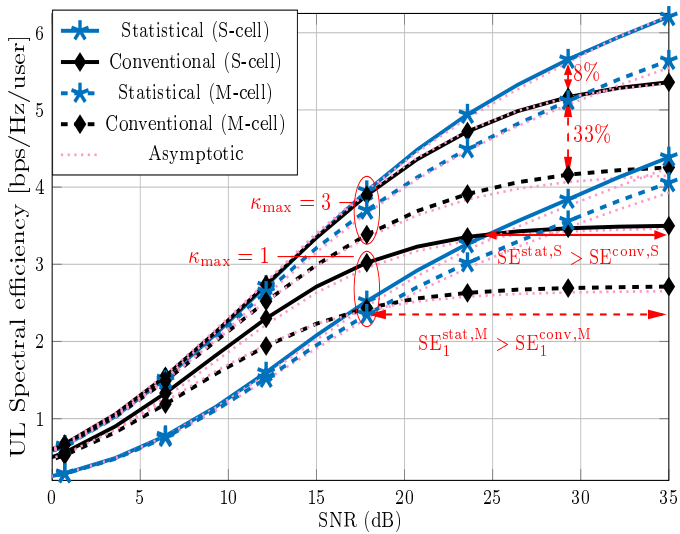}
\end{minipage}\hspace{-1em}
\begin{minipage}[c]{0.42\linewidth}
\vspace{-1.5em}
\caption{\small UL SEs of conventional and statistical combining in both single-cell $(L=1)$ and multi-cell $(L=3)$ settings, for different levels of Rician factor, and $K=20$ cell-edge users. Solid and dashed lines (and arrows) correspond to the single-cell and multi-cell cases, respectively. Dotted lines represent the asymptotic approximations given in Theorems \ref{th: MMSE hat} - \ref{th:multi statistical rate}.}\label{fig:Multi_Single_comp}
\end{minipage}
\vspace{-2 em}
\end{figure}
%%%%%%%%%%%%%%%%%%%
\par Next, to highlight the impact of inter-cell interference on the performance of the receivers, we consider cell $1$ from the network setup in Fig.\ref{fig:Netset} as a cell of interest, and propose to compare its achievable UL SE in the cases where it is deployed: 
\hspace{3em}
\begin{itemize}
\item[(1)] in the multi-cell setting of Fig.\ref{fig:Netset},
\item[(2)] in a single-cell setting having the same system and channel parameters as in (1). 
\end{itemize}
Accordingly, we represent in Fig.\ref{fig:Multi_Single_comp} the SEs corresponding to these cases for different levels of $\kappa_{\rm max}$. Dashed and  solid lines (and arrows) correspond to cases (1) and (2), respectively.  As indicated for $\kappa_{\rm max}= 2$, in the multi-cell setting, the statistical combining achieves a $33\%$ SE gain over the conventional one; whereas in the single-cell case, the observed increase is by $8\%$ only. As for smaller Rician factors ($i.e.$, $\kappa_{\rm max}=1$), we can see that in the multi-cell plot, the statistical receiver outperforms the conventional one for a lower SNR (starting $17$ dB) than it is the case for the single-cell scheme, wherein this is only achieved for SNRs above $25$dB.  
% It is worth mentioning that the SEs plotted in Fig.\ref{fig:Multi_Single_comp} are different from those in Fig.\ref{fig:Conv_Stat} since we consider here  
Consequently, Fig.\ref{fig:Multi_Single_comp} validates that, in the infinite antenna-limit, compared to conventional combining, employing the statistical receiver works even better in a multi-cell network. As demonstrated in Section\ref{sec:multi}, this is explained by the fact that since it mitigates the inter-cell interference, this processing technique actually engenders a similar multi-cell SE, ${\rm SE}^{\rm stat, M}$, to when used in a single-cell system, ${\rm SE}^{\rm stat,S}_k$. To summarize, Fig.\ref{fig:Multi_Single_comp} asserts that for the same cell $j$: $\overline{\rm SE}^{\rm stat, S}_k-\overline{\rm SE}^{\rm stat,M}_{jk}\asto 0$, whereas, due to pilot contamination, $\overline{\rm SE}^{\rm conv , S}_k\geq \overline{\rm SE}^{\rm conv, M}_{jk}$. Nonetheless, note that the gap observed in Fig.\ref{fig:Multi_Single_comp} between ${\rm SE}^{\rm stat, S}_k$ and ${\rm SE}^{\rm stat, S}_k$ is due to the finite system dimension considered in the simulations which will be reduced as $N$ grows larger.
\vspace{-0.5em}
\section{Conclusion}\label{sec:Conclusion}
We studied in this work the UL performances of single and multi-cell massive MIMO systems underlying spatially correlated Rician channels, with the assumption of imperfect channel estimates. Considering the large-antenna limit, we derived closed-form approximations of the spectral efficiencies achieved by the LMMSE conventional receiver and proposed a novel statistical combining scheme. For the former, the approximations were exploited to determine an explicit expression of the optimal number of training symbols which was shown to be particularly important for small Rician factors. Conversely, the study reveals that, in LoS-prevailing environments, investing in longer training sequences to enhance the SE is ineffective. This result, led us to propose the statistical receiver that is more beneficial for systems with strong LoS components. The multi-cell analysis unveiled that conventional processing is limited by pilot contamination, even under favorable propagation; yet, it demonstrated that stronger LoS signals reduce this correlated interference. On another note, the asymptotic derivations indicated that statistical combining allows to mitigate inter-cell interference, and as such, outperforms the conventional receiver in a multi-cell system to an even higher extent.  Finally, the approximations given in this work can be applied for realistic scenarios involving different correlation matrix models, Rician factors, CSI errors. In essence, they provide a general framework that can be harnessed to perform further analysis of similar networks.}
 %On another note, the techniques in this study can be extended to the downlink of correlated Rician fading channels in large-scale MIMO systems. Additionally, the approximations given in this work can be applied for realistic scenarios involving different correlation matrix models, Rician factors, CSI errors. In essence, they provide a general framework that can be harnessed to perform further analysis of similar networks, without resorting to heavy Monte Carlo simulations.
%% The proposed receiver is statistical as it is obtained though the maximization of a deterministic equivalent of the sum-rate. 

%\newpage
%%
%
%All the matrices in what follows are $\in \mathbb{C}^{N\times N}$. $\boldsymbol{\Theta}_{j\ell k}$ is a correlation matrix. $\beta_{j\ell k}$, $\kappa_{jk}$ are positive scalars. 
%
%Let: 
%
%\begin{align*}
%&{\bf R}_{j \ell k} ={\frac{\beta_{j \ell k} }{1+\kappa_{jk}\delta_{j\ell}}}\boldsymbol{\Theta}_{j \ell k}, \\
%&\bPhi_{jk} =\left(\sum_{\ell' =1}^{L}{\bf R}_{j \ell' k} + \frac{1}{\tau \rho_{tr}}\bI_N\right)^{-1},\\
%&\tilde{\bR}_{jjk}= {\bR}_{jjk}\bPhi_{jk}{\bR}_{jjk}
%\end{align*}
%
%with $\delta_{j\ell}$ Kronecker delta : $\delta_{j\ell}=1$ iif $\ell=j$, and $\delta_{j\ell}=0$ otherwise. 
%
%Note: ${\bR}_{jjk}$, $\tilde{\bR}_{jjk}$ and ${\bR}_{jjk}-\tilde{\bR}_{jjk}$ are Covariance matrices $=>$ all PSD.
%
%
%%Task: Let $f(\kappa_{jk})=\tr({\sum_{\ell\neq j}^L}\bR_{j\ell k}\bPhi_{jk}{\bR}_{jj k})$. Prove that:
%
%Prove that :
%
%\begin{equation}
%\tr\left({\sum_{\ell\neq j}^L}\bR_{j\ell k} \bPhi_{jk} (\bR_{jjk} -\tilde{\bR}_{jjk})\right) \geq 0. 
%\end{equation}
%
%%
%\pagebreak

\vspace{-1 em}
\begin{appendices}
 \section{Proof of Theorem \ref{th: MMSE hat}}\label{app: conv LMMSE}
We demonstrate in this section the results of Theorem \ref{th: MMSE hat}. As the derivations rely most often on the same arguments, we  mention the pertinent steps to derive the asymptotic approximation \eqref{eq:R MMSE DE}. 
First, define $\tilde{\bf Q}=\left(\frac{1}{N}\wbHh {\bf Z} \wbH+\frac{1}{\rho_d}\bI_K\right)^{-1}
$, with ${\bf Z}^{-1}= \frac{\rho_d}{N}\sum_{i=1}^{K} \left({\bf R}_i-\tilde{\bf R}_i\right)+ \bI_N$. 
Second, using the Woodbury matrix identity enables to express all the signals constituting ${\rm SE}^{\rm conv,S}_k$ in terms of the elements of matrix $\tilde{\bQ}$. For instance, the signal term $\vert{\bf g}_k\herm\wbh_k\vert^2$, can be written as:
{\small \begin{equation}\label{eq:sig simp}
\left| \frac{1}{N} \wbhh_k \left(\frac{\wbH\wbH\herm}{N}+ \frac{1}{N}\sum_{i=1}^{K} \left({\bf R}_i-\tilde{\bf R}_i\right)+\frac{1}{\rho_d}\bI_N\right)^{-1}\wbh_k\right|^2= \left|1-\frac{1}{\rho_d} [\tilde{\bf Q}]_{kk}\right|^2.
\end{equation}}  
Accordingly, ${\rm SE}^{\rm conv,S}_k$ in \eqref{eq:sum rate hat} can be rewritten as follows: 
{\small \begin{align}
&{\rm SE}^{\rm conv,S}_k =  \notag \\ &
\left(1-\frac{\tau}{T}\right)\log\left(1+ 
\frac{\left|1- \frac{1}{\rho_d}\left[\tilde{\bQ}\right]_{kk}\right|^2}{ {\sum_{\substack{i=1\\ i\neq k}}^{K}}\left|\frac{1}{\rho_d}\left[\tilde{\bQ}\right]_{ki}\right|^2
+ \frac{1}{N^2}\sum_{i=1}^{K}\tilde{\bq}_k\herm \wbHh{\bf \bxi}_i{\bf \bxi}_i\herm\wbH\tilde{\bq}_k
+ \frac{1}{\rho_d}\left(\left[\tilde{\bQ}\right]_{kk} -\frac{1}{\rho_d} \left[\tilde{\bQ}^2\right]_{kk}\right)
 }\right). \label{eq:sum-rate Q tilde}
\end{align}}
In fact, putting the spectral efficiency ${\rm SE}^{\rm conv,S}_k$ in this form \eqref{eq:sum-rate Q tilde} facilitates the derivation of the deterministic equivalent of this latter since we can simply use the LLN as follows: 
%\begin{prop*}
\textit{\begin{itemize}
\item  Under assumption \ref{ass:asymptotic}, the LLN allows us to put $\frac{1}{N} \left[{\wbHh\wbH}\right]_{ij}-\frac{1}{N}\mathbb{E} \left[\wbhh_i\wbh_j\right] \asto 0. 
$\\
Therefore, under Assumption \ref{ass:asymptotic}, using the continuous mapping theorem \cite{book-Prob95}, we have :
 \begin{equation}\label{eq:QtildeQ}
 [\tilde{\bQ}]_{ij} - \left[{\bQ}\right]_{ij} \asto 0,
 \end{equation}\vspace{-0.7em}where the matrix $\bQ$ is given in \eqref{eq:Q_hat}.
\end{itemize}}
Thus, a direct application of \eqref{eq:QtildeQ} with the continuous mapping theorem \cite{book-Prob95} enables us to find asymptotic approximations of most of the terms in \eqref{eq:sum-rate Q tilde}. For instance, for the signal term, we find: 
 $\left|1- \frac{1}{\rho_d}\left[\tilde{\bQ}\right]_{kk}\right|^2 - \left|1-\frac{1}{\rho_d} [\bQ]_{kk}\right|^2\asto 0$. Likewise, the same steps allow to derive approximations of the intra-cell interference term and processed noise.
 { As to the estimation error term $\mathcal{E}_k=\frac{1}{N^2}\sum_{i=1}^{K}\tilde{\bq}_k\herm \wbHh{\bf \bxi}_i{\bf \bxi}_i\herm\wbH\tilde{\bq}_k$, we mainly adopt the same reasoning except for the following step. Indeed, in order to find a deterministic equivalent for $\mathcal{E}_k$, we first exploit the orthogonality property of LMMSE channel estimation by observing that $\forall\{k,i\}$, $\wbh_k$ and ${\bf \bxi}_i$ are independent. %This fact allows us to use the convergence of quadratic forms lemma \cite{eigenvaluesoutside-silverstein2009}, and put: 
% \begin{equation}
% \mathcal{E}_k-\frac{1}{N^2}\sum_{i=1}^{K} \bq_k\herm \wbHh \left(\frac{1}{\tau \rho_{tr}}{\bf R}_i\bPhi_i\right) \wbH \bq_k\asto 0.
% \end{equation}
After that, applying the convergence of quadratic forms lemma \cite{eigenvaluesoutside-silverstein2009} yields: 
{\small \begin{equation}
\mathcal{E}_k-\frac{1}{N^2}\sum_{i=1}^{K} \bq_k\herm \left(\overline{\bH}\herm\left(\frac{1}{\tau \rho_{tr}}{\bf R}_i\bPhi_i\right)\overline{\bH} + diag\left\{ \tr\left(\tilde{\bf R}_\ell\frac{1}{\tau \rho_{tr}}{\bf R}_i\bPhi_i\right)\right\}_{\ell=1}^K \right)\bq_k\asto 0.
\end{equation}}}
Finally, putting all the above terms together yields the asymptotic approximation of the spectral efficiency in Theorem \ref{th: MMSE hat} and as such, concludes the proof. \vspace{-1.1em}
\section{Proof of Theorem \ref{th:optim}}\label{app:tau_opt}
%First, note that the behavior of $\overline{\rm SE}^{\rm conv}$ with respect to $tau$ is concave   
Denote $\partial_\tau F$ and $\partial^2_\tau F $ as the first and second derivatives of any function $F(\tau)$ with respect to $\tau$. To prove the results of Theorem \ref{th:optim}, we use the following approach: 
First, we show that $\forall \tau $, $\partial_\tau{\rm SE}^{\rm conv}$ is monotonically decreasing and that $\exists \tau_0 $, such that $\partial_\tau{\rm SE}^{\rm conv}\vert_{\tau=\tau_0}=0$. Therefore, $\tau_0$ is unique and $\tau_0=\underset{\tau} {\rm argmax} \ \overline{\rm SE}^{\rm conv}$. Accordingly, if this step is verified, considering the constraint in \eqref{P:tau opt}, finding $\tau^*$ is simply obtained as:
\begin{itemize}
\item If $\tau_0 \leq K$, then %$\partial_\tau{\rm SE}^{\rm conv}\leq 0$, $\forall \tau\in[K,T[$ and thus, increasing $\tau$ will decrease the achievable UL rate. Therefore, under the constraint of \eqref{P:tau opt},
 $\tau^*=K$, which is depicted by the first solution \eqref{eq:K tau opti}.
\item  On the other hand, if $\tau_0 \in [K,T[$, then $\tau^*=\tau_0$ which is represented by the solution \eqref{eq:tau_opt fix}.
\end{itemize} 
\paragraph*{\bf Proof that $\forall \tau$, $\partial_\tau{\rm SE}^{\rm conv}$ is monotonically decreasing}
 To establish this, a sufficient condition is to have: $ \forall \tau \in [K,T[$, $\partial^2_\tau {\rm SE}^{\rm conv,S}_k < 0$. 
%Second, if $\partial_\tau{\rm SE}^{\rm conv}\vert_{\tau=K} \leq 0$ and  $ \forall \tau \in [K,T[$, $\partial^2_\tau {\rm SE}^{\rm conv,S}_k < 0$ hold.
In this line, using \eqref{eq:R MMSE DE simplified}, we have :  
{ \begin{align}
&\partial_\tau \overline{\rm SE}^{\rm conv} = \frac{1}{K} \sum_{k=1}^{K} \left[\left(1-\frac{\tau}{T}\right) \frac{\gamma'_k (\tau)}{1+\gamma_k(\tau)}-\frac{1}{T}\log\left(1+\gamma_k (\tau)\right)\right],\label{eq:R_prime} \\
&\partial^2_\tau \overline{\rm SE}^{\rm conv} = \frac{1}{K} \sum_{k=1}^{K} \left[-\frac{2}{T} \frac{\gamma'_k (\tau)}{1+\gamma_k(\tau)}+\left(1-\frac{\tau}{T}\right) \frac{\gamma''_k(\tau) \left(1+\gamma_k(\tau)\right) - \left(\gamma'_k(\tau)\right)^2}{\left(1+\gamma_k(\tau)\right)^2}\right],\label{eq:R_seconde} 
\end{align}}
with $\gamma'_k (\tau)$ and $\gamma''_k (\tau)$ given by : \vspace{-0.5em}
{\small \begin{align}
&\gamma'_k (\tau) = \rho_d  \frac{\bq\herm_{k} \overline{\bf D}_2 \bq_k}{ \left(\left[\bQ\right]_{kk}\right)^2}, \label{eq:gamma_prime}\\
%&\gamma''_k (\tau) = \rho_d  \frac{-2 \bq_k\herm \left(\overline{\bf D}_2 \bQ \overline{\bf D}_2 - \overline{\bf D}_3 \right)\bq_k \left[\bQ\right]_{kk} +2 \left(\bq\herm_{k} \overline{\bf D}_2 \bq_k\right)^2 }{ \left(\left[\bQ\right]_{kk}\right)^3} \\
&\gamma''_k (\tau) =-2 \rho_d  \frac{ \bq_k\herm \overline{\bf D}_2 \left(\bQ \left[\bQ\right]_{kk} -\bq_k\bq_k\herm \right)\overline{\bf D}_2 \bq_k - \left[\bQ\right]_{kk} \bq\herm \overline{\bf D}_3 \bq_k}{ \left(\left[\bQ\right]_{kk}\right)^3}
\label{eq:gamma_seconde},
\end{align}}
 and 
% {\color{red}  equations (15-21) correctes 
% \begin{align}
%& \partial_\tau \left[\bQ\right]_{kk} =- \bq\herm_{k} \overline{\bf D}_2 \bq_k,\label{eq:Q_prime} \\
%&\partial^2_\tau \left[\bQ\right]_{kk} = 2 \bq_k\herm \left(\overline{\bf D}_2 \bQ \overline{\bf D}_2 - \overline{\bf D}_3 \right)\bq_k.\label{eq:Q_seconde} 
% \end{align}
% }
% \begin{align}
% \overline{\bf D}=\frac{1}{\rho_{tr} \tau^2}diag\left\{\frac{1}{N}\tr\left({\bf R}_\ell^2 \bPhi^2_\ell \right) \right\}_{\ell =1}^{K}. \\
% \partial_\tau \overline{\bf D} = \frac{-2}{\rho\tau^3} diag\left\{\frac{1}{N}\tr\left({\bf R}_\ell^3 \bPhi^3_\ell \right) \right\}_{\ell =1}^{K}.  
% \end{align}
% \begin{equation}
$ \overline{\bf D}_\alpha=\frac{\left(-1\right)^\alpha}{\rho_{tr} \tau^\alpha}diag\left\{\frac{1}{N}\tr\left({\bf R}_\ell^\alpha \bPhi^\alpha_\ell \right) \right\}_{\ell =1}^{K}$,%\label{eq:D_bar_alpha} $
 %\end{equation}
with $\alpha$ being an integer. Note that $\overline{\bf D}_\alpha$ is a positive semi-definite matrix for all even values of $\alpha$, and negative semi-definite otherwise. %for all odd values of $\alpha$.  
 With this in mind, from \eqref{eq:gamma_prime}, we can see that $\gamma'_k\left(\tau\right) \geq 0$, $\forall \tau\in\left[K,T\right[$. Plus, since $\forall \tau$, $\gamma_k(\tau)$ is evidently positive, it can be seen from \eqref{eq:R_seconde} that $\gamma''_k (\tau)\leq 0$, $\forall \tau\in\left[K,T\right[$ is a sufficient condition to obtain : $\partial^2_\tau \overline{\rm SE}^{\rm conv}\leq0$. With this in mind, using the fact %\footnote{This is equivalent to ${\bf A} \preccurlyeq \tr\left({\bf A}\right)\bI_M$ for any positive definite matrix $\bA\in\mathbb{C}^{M\times M}$} 
 that ${\bf a}\herm{\bf a} \bI_M - {\bf a}{\bf a}\herm$ is a positive semi-definite matrix, we can easily show that $\gamma''_k(\tau)\leq 0$, $\forall \tau$. This concludes the proof that  $\partial_\tau{\rm SE}^{\rm conv}$ is monotonically decreasing with respect to $\tau$, and validates the results given in Theorem \ref{th:optim}.
{ \vspace{-1.2em}
\section{Proof of Corollary \ref{cor:kappa_up}}\label{app:kappa_up}
%To solve \eqref{P:rates comp}, a sufficient condition is to find, $\forall k$, $\kappa_k$ that verifies : 
%\begin{equation}\label{eq:P33}
%\left(1-\frac{\tau}{T}\right) \log\left[1+  \rho_d\left(\frac{1}{N}\tr \tilde{\bf R}_k + \frac{\beta_k \kappa_k}{1+\kappa_k}\right)\right] \leq  \log\left(1+ \rho_d  \frac{\beta_k \kappa_k}{1+\kappa_k}\right). 
%\end{equation}
For simplicity, note that the index ``$k$'' will be dropped in the sequel. Accordingly, denoting $\alpha_i$ the $i-$th eigenvalue of $\boldsymbol{\Theta}$, our objective is to find $\kappa$ such that : 
\begin{align*}\label{ineq:rates comp}
\left(1-\frac{\tau}{T}\right)&\log\left[1+\frac{\beta\rho_d}{1+\kappa}\left(\frac{1}{N} \sum_{i=1}^{N}\frac{{\alpha_i}^2}{\alpha_i+\frac{1+\kappa}{\rho_{tr} \beta \tau}} + \kappa\right)\right]  %\\ & 
\leq \tag{P3'} \log\left(1+\rho_d\beta \frac{\kappa}{1+\kappa}\right). 
\end{align*}
Since ``$\log$" is an increasing function and $\alpha_i \geq 0$, $\forall i$, we consider the upper bound: 
$\frac{\alpha_i^2}{\alpha_i+\frac{1+\kappa}{\tau\rho_{tr}}} \leq \alpha_i, $ for all positive values of $\kappa$, $\tau$ and $\rho_{tr}$. Therefore, \eqref{ineq:rates comp} is satisfied whenever $\kappa$ verifies:
\begin{equation}\label{eq:kappa_cond_appendix}
\left(1-\frac{\tau}{T}\right)\log\left[1+\frac{\beta\rho_d}{1+\kappa}\left(\frac{1}{N} \tr\boldsymbol{\Theta} + \kappa\right)\right] \leq \log\left(1+\rho_d\beta \frac{\kappa}{1+\kappa}\right).
\end{equation}
 Applying $``\exp"$ on both sides of \eqref{eq:kappa_cond_appendix} yields the lower bound $\kappa\geq f\left(\kappa\right)$, with: 
 \begin{equation}\label{eq:f_kappa}
 f\left(\kappa\right) =
-1+\frac{ \frac{1}{N}\tr \boldsymbol{\Theta}}{-\frac{\kappa}{\kappa+1} +\frac{1}{\beta\rho_d}\left(-1+\left(1+\beta \rho_d \frac{\kappa}{1+\kappa}\right)^{\frac{T}{T-\tau}}\right)}.
 \end{equation}
%It can be shown that $f\left(\kappa\right)$ is an increasing function of $\kappa$.
  %Thus, finding $\kappa \geq {\underset{\kappa} {\rm \bf min}} f\left(\kappa\right)$ verifies the former inequality.%\footnote{\label{note1} If $f(x)$ is a decreasing function of $x \in \left[a,b\right]$, $\forall \ x$, $f(x) \leq f\left(a \right)$.}. 
 Next, we use the following result : \textit{Let $x$, $\alpha$ $\in \mathbb{R}$, and consider $y(x)$, a positive increasing function of all $x\geq 0$ with $y(0)=0$. If $\alpha > 1$, the following inequality holds:} \vspace{-1em}
\begin{equation}\label{eq:ineq}
\left(1+y(x)\right)^\alpha \geq \left(1+\alpha y(x)\right),
\end{equation}
 \begin{rem}
This result can be proved by showing that the function $g(x)=\left(1+y(x)\right)^\alpha - \left(1+\alpha y(x)\right)$ verifies, $\forall x\geq 0$:  $g(0)=0$ and $g'(x)\geq 0$, thus yielding: $g(x)\geq 0$, $\forall x\geq 0$.
 \end{rem}
 Accordingly, applying \eqref{eq:ineq} on \eqref{eq:f_kappa} allows us to obtain the following condition on $\kappa$:  
$$\kappa \geq -1+ \frac{\tr\boldsymbol{\Theta}}{N} \frac{\left(T-\tau\right)\left(1+\kappa\right)}{\tau\kappa}. $$
This latter admits the solution: $\kappa\geq \frac{\tr \boldsymbol{\Theta}}{N} \frac{T-\tau}{\tau}.$ Finally, since the function $\frac{T-\tau}{\tau}$ is decreasing in $\tau$, plus the fact that $\tau\in\left[K,T\right)$, we can simply consider the lower bound given in \eqref{ineq:kappa cond2}, therefore concluding the proof.
}
{
\vspace{-2em}
\section{Proof of Theorem \ref{th:multi-MMSE}}\label{app:MMSE multicell proof}
The same steps and arguments given in appendix \ref{app: conv LMMSE} can be used to find 
the asymptotic approximation $\overline{\rm SE}^{\rm conv,M}_{jk}$ in Theorem \ref{th:multi-MMSE}, and are thus omitted due to space limitations. Nevertheless, the main difference lies in the inter-cell interference term, where it is imperative to take into account the correlation between the estimates and the interfering channels that share the same pilot, s.t, $\forall \ell\neq j$ : $\frac{1}{N}\mathbb{E} \left[\wbhh_{jji} \wbh_{j\ell i}\right]- \frac{1}{N}\tr(\bR_{j\ell i}\bPhi_{jk}{\bR}_{jj i})\asto 0$.
\vspace{-1em}
\section{Proof that pilot contamination is decreasing with respest to $\kappa_{jk}$}\label{app: pilot vs kappa}
First we need the following preliminary results. \vspace{-1em}
\begin{lemma}\cite{MatrixA-Horn2013}\label{lem:PSD}
For any positive semi-definite (PSD) $N\times N$ matrices $\bA$ and $\bB$, the matrices $\bA \bB \bA$, $\bB \bA \bB$ and $\bA+\bB$ are positive semi-definite, and $\tr\left(\bA\bB\right)\geq 0 $. Plus, if $\bA\bB=\bB\bA$, then $\bA\bB$ is also PSD. Finally, if $\bA$ is positive definite (PD), then $\bA^{-1}$ is also PD.  
\end{lemma}
Second, let $=\frac{1}{N}\tr({\sum_{\ell\neq j}^L}\bR_{j\ell k}\bPhi_{jk}{\bR}_{jj k})$. A straightforward differentiation of $f(\kappa_{jk})$ yields:\vspace{-1.2em}
\begin{equation}
f'(\kappa_{jk})=\frac{-1}{1+\kappa_{jk}}\frac{1}{N}\tr\left({\sum_{\ell\neq j}^L}\bR_{j\ell k} \bPhi_{jk} (\bR_{jjk} -\tilde{\bR}_{jjk})\right). 
\end{equation}
Accordingly,  we prove in what follows that $\tr\left({\sum_{\ell\neq j}^L}\bR_{j\ell k} \bPhi_{jk} (\bR_{jjk} -\tilde{\bR}_{jjk})\right)\geq 0$ which eventually leads to $f'(\kappa_{jk})\leq 0$. 
In this line, we assume that the correlation matrix $\bR_{j j k}$, is positive definite. That is, we add the `perturbation' $\epsilon \bI_N$, $\epsilon>0$ as follows: 
denote $\underline{\bR}_{jjk}=\bR_{jjk} + \epsilon \bI_N$, and $\underline{\bPhi}_{jk}=\left(\underline{\bR}_{jjk}+ \sum_{\ell\neq j} \bR_{j\ell k}+ \frac{1}{\rho_{tr}}\bI_N\right)^{-1}$. 
Therefore, under this assumption, let $g(\kappa_{jk}, \epsilon) = \tr\left({\sum_{\ell\neq j}^L}{\bR}_{j\ell k} \underline{\bPhi}_{jk} (\underline{\bR}_{jjk} -\underline{\bR}_{jjk} \underline{\bPhi}_{jk} \underline{\bR}_{jj k} )\right)$. Note that $g(\kappa_{jk}, 0)=\tr\left({\sum_{\ell\neq j}^L}\bR_{j\ell k} \bPhi_{jk} (\bR_{jjk} -\tilde{\bR}_{jjk})\right)$. Plus, since $g(\kappa_{jk},\epsilon)$ is continuous, if $g(\kappa_{jk},\epsilon)>0 $, $\forall \epsilon >0$, then $\displaystyle{\lim_{\epsilon\rightarrow 0}} g(\kappa_{jk}, \epsilon) \geq 0$.% Proving this statement will conclude the proof by observing that  
\begin{itemize}
\item $\forall \epsilon>0$, let: $\bB_{jk}={\sum_{\ell\neq j}^L}{\bR}_{j\ell k}$, $\bA_{jk}=\left(\bB_{jk} + \frac{1}{\tau \rho_{tr}}\bI_N\right)^{-1}$. 
\end{itemize}
Therefore $\underline{\bPhi}_{jk}= \left(\underline{\bR}_{jjk}+\bA_{jk}^{-1}\right)^{-1}$. Now, using the Woodbury matrix identity, we can put:  $\underline{\bR}_{jjk} \underline{\bPhi}_{jk} \underline{\bR}_{jj k}= \left(\underline{\bR}_{jjk}^{-1} + \bA_{jk}\right)^{-1}$.Therefore: \vspace{-1em}
\begin{align}
g(\kappa_{jk}, \epsilon)& =  \tr\left( \bB_{jk} \left(\underline{\bR}_{jjk}+ \bA_{jk}^{-1}\right)^{-1} \left(\underline{\bR}_{jjk}^{-1}+ \bA_{jk}\right)^{-1} \right) \notag \\
& = \tr\bigg( \bB_{jk} \bA_{jk}  \underbrace{\big( 2 \bA_{jk} + \underline{\bR}_{jjk}^{-1} + \bA_{jk} \underline{\bR}_{jjk}\bA_{jk}\big)^{-1}}_\text{PD due to lemma \ref{lem:PSD}} \bigg). \label{eq:544}
\end{align}
Plus, observing that $\bB_{jk} \bA_{jk}= \bA_{jk} \bB_{jk} $, the product $\bB_{jk} \bA_{jk}$ is also a PSD matrix, based on lemma \ref{lem:PSD}. Consequently, we find that \eqref{eq:544} amounts to the trace of the product of two PSD matrices which is always a positive quantity, thus concluding the proof.}
\end{appendices}
\vspace{-0.5em}{
\bibliographystyle{IEEEtran} 	
\vspace{-0.5em}
% (uses file "plain.bst")
\bibliography{ref}}
\end{document}